\documentclass[10pt]{article}
\usepackage[preprint]{tmlr}
\usepackage{amsmath,amssymb,amsthm,mathtools}
\usepackage{microtype}
\usepackage{booktabs}
\usepackage{array}
\usepackage{graphicx}
\usepackage{longtable}
\usepackage[font=small]{caption}
\usepackage{hyperref}
\usepackage{url}

\newtheorem{theorem}{Theorem}
\newtheorem{lemma}{Lemma}
\newtheorem{corollary}{Corollary}
\theoremstyle{definition}
\newtheorem{assumption}{Assumption}
\newtheorem{remark}{Remark}

\newcommand{\1}{\mathbf{1}}
\newcommand{\E}{\mathbb{E}}
\newcommand{\Prob}{\mathbb{P}}
\newcommand{\Xcal}{\mathcal{X}}
\newcommand{\Ical}{\mathcal{I}}
\newcommand{\wstar}{w^{\star}}
\newcommand{\Ghat}{\widehat{G}}
\newcommand{\what}{\widehat{w}}
\newcommand{\qhat}{\widehat{q}}
\newcommand{\tstar}{t^{\star}}

\newcommand{\supmat}{Appendix~\ref{sec:s-further}}
\newcommand{\Supmat}{Appendix~\ref{sec:s-further}}

\title{When does conformal calibration need censoring weights?\\ Cause-of-failure
prediction sets under competing risks}

\author{\name Sunny Yang \email guangy4@illinois.edu\\
\addr University of Illinois Urbana-Champaign
\AND
\name Weiyan Zhao \\
\addr Northumbria University}

\begin{document}
\maketitle

\begin{abstract}
Split conformal prediction sets for competing-risks labels at a fixed horizon
require calibration labels that right censoring can leave unobserved.
Complete-case calibration guarantees coverage for the label-complete
subpopulation, but its population coverage can deviate in either direction, even
at the true class probabilities. We study how selection changes the score
distribution near the population quantile. At the true score in our main
simulation family, with independent draws, complete-case calibration covers
0.8723 at a nominal 0.900 when 22\% of subjects are event-free at the horizon. In
4 of 48 further designs using per-draw normalisation, estimating cause incidences
as one minus the exponential of the negative of the cause-specific Nelson--Aalen cumulative hazards, with the
event-free probability as the clipped and renormalised remainder, puts
complete-case coverage at least four standard errors below nominal. In these
designs, weights from a correctly specified censoring model keep coverage near
nominal. Extending the argument of Yi et al.\ (2025) to a cause label, we
establish a finite-sample coverage lower bound with an explicit penalty for
censoring-model error. Misspecifying the censoring model can also lead to
under-coverage.

\end{abstract}

\section{Introduction}
\label{sec:intro}

For a competing-risks model, a natural prediction at a fixed horizon is a set of
causes, possibly including event-free status, that contains the realised label at
the horizon for at least $1-\alpha$ of subjects. Split conformal prediction gives
such a set from any fitted model when the calibration labels are observed. Under
right censoring they are not, since a subject who leaves follow-up before the
horizon without an event has no label.

Complete-case calibration keeps the subjects whose label is determined and drops
the rest. It selects on the outcome, because a subject who fails early is more
likely to have been followed long enough than one who survives past the horizon,
and the exchangeability that conformal prediction relies on is lost. Bostr\"om et
al.\ (2019) use this procedure for a binary event-by-horizon label and state its
guarantee for the subjects not censored before the horizon. Bostr\"om et al.\ (2023)
keep the censored subjects, with a Kaplan--Meier estimate within each Mondrian
category, to give a predictive distribution for the event time. Reports on the
population disagree. Dang et al.\ (2026) find complete-case conformal
under-covering, attribute it to the censoring rate and recommend a switch above
40\%. Farina et al.\ (2026) find their complete-case comparator over-covering in
every setting they run, but their target is a lower prediction bound on the event
time, for which selection toward short observed times is conservative by
construction. Neither account applies to a label at a horizon.

We characterise when complete-case calibration is safe for a cause-of-failure
label and when weighting the calibration sample by the inverse probability of
remaining uncensored repairs it. The paper makes three contributions.
\begin{enumerate}
\item \emph{The complete-case error can take either sign, even at the true
score.} Complete-case calibration takes its threshold from the score
distribution of the selected subjects, and we study how this selection changes
coverage over the population. At fixed covariates, subjects event-free at the horizon have the lowest
selection probability, so the event-free share, error in the fitted event-free
probability and the dependence of censoring on the covariates all shift the
selected score distribution relative to the population's. The direction of the
mean score shift indicates the sign of the error but does not determine it
(Appendix~\ref{sec:suppreading}). In our main simulation family at the true score,
the unweighted procedure over-covers at an event-free share of $0.544$ and
under-covers at $0.224$, where it covers $0.8723$ at a nominal $0.900$ over $200$
replications with independent draws (Section~\ref{sec:correctfit}). In a
two-region law where most subjects are event-free, the true score still
under-covers (Section~\ref{sec:wrongG}). A fit that understates the event-free
probability lowers coverage, censoring amplifies both effects, and the observed
crossing point moves with the censoring rate, so no threshold on the censoring
rate locates the failure.
Estimating each cause's incidence as one minus a cause-specific survival curve,
with the event-free probability as the remainder, puts the unweighted procedure
at least four standard errors below nominal in 4 of 48 designs
(Section~\ref{sec:dials}).

\item \emph{Censoring weights give a coverage guarantee for any score when the
censoring model is correct.} In the studied designs with covariate-free
censoring, where the marginal Kaplan--Meier estimator is correctly specified, the
weighted procedure stays within $0.012$ of nominal across event-free shares, misspecification levels and
miscoverage levels. When censoring depends on the covariates, marginal
Kaplan--Meier weights fail with the unweighted procedure, at $0.8294$ in a
two-region law where region-stratified weights cover $0.9496$, and miss in either
direction under a misspecified score, while weights from the true censoring model
hold (Section~\ref{sec:wrongG}). Table~\ref{tab:twobytwo} summarises the four
cases.

\item \emph{A finite-sample coverage floor for the weighted procedure.}
Lemma~\ref{lem:we} establishes weighted exchangeability between the
label-complete calibration sample and the test subject, and
Theorem~\ref{thm:main} gives a floor of $1-\alpha-\Delta_w$ with $\Delta_w$
explicit in the censoring-model error. The argument is the factorisation of Yi et
al.\ (2025, Appendix A) applied to the selection event $\{C\ge\min(T,\tstar)\}$ in
place of $\{T\le C\}$, with the normalising constant $\kappa$ kept explicit and the
random size of the complete-case pool accounted for. The cause label changes one
step. A hypothesised time determines the test subject's weight, but a hypothesised
cause leaves $\min(T,\tstar)$ unknown, so the procedure uses a dominating
substitute for that weight, the atom.
\end{enumerate}

\begin{table}[t]
\centering
\caption{Behaviour of the two procedures in the studied simulation designs, by
which model is correct. Unweighted is
split conformal on the label-complete subjects. Weighted uses
inverse-probability-of-censoring weights from the censoring model named in the
column.}
\label{tab:twobytwo}
\small
\begin{tabular}{@{}l>{\raggedright\arraybackslash}p{0.38\textwidth}>{\raggedright\arraybackslash}p{0.38\textwidth}@{}}
\toprule
 & censoring model correct & censoring model misspecified \\
\midrule
score correct &
Unweighted over- or under-covers, depending on where selection moves the score
distribution, in our main family on the event-free share. Weighted stays near
nominal. (Section~\ref{sec:correctfit}) &
Unweighted can fail, and so can weighted with a marginal Kaplan--Meier. Weights
from the correct model hold. (Section~\ref{sec:wrongG}) \\
\addlinespace
score misspecified &
Unweighted falls below nominal once the event-free probability is understated
enough, at every censoring rate. Weighted stays within $0.012$ of nominal.
(Section~\ref{sec:dials}) &
Unweighted and marginal-Kaplan--Meier weights miss in either direction. Weights
from the true model hold. (Section~\ref{sec:wrongG}) \\
\bottomrule
\end{tabular}
\end{table}

The class-conditional procedure, which restores coverage of a rare cause that a
pooled threshold misses, is in Section~\ref{sec:percause}, with a registry cohort
of $238{,}057$ subjects.

\paragraph{Related work.} Conformal prediction under covariate shift is due to
Tibshirani et al.\ (2019), the penalty for an estimated shift ratio to Lei and
Cand\`es (2021), and class-conditional, or Mondrian, conformal prediction to Vovk
et al. Podkopaev and Ramdas (2021) reweight conformal classification under label
shift. Their weight $q(y)/p(y)$ depends on the class index alone, so a
hypothesised label determines the test weight, and their calibration set is fully
labelled. Our weight $1/G(\min(T,\tstar)\mid X)$ depends on the subject's event
time, which a hypothesised cause leaves unknown. With estimated ratios their
statement becomes asymptotic, whereas Theorem~\ref{thm:main} keeps a finite-sample
floor at an explicit cost $\Delta_w$.

For survival outcomes, the closest in mechanism is Yi et al.\ (2025), who
weight by the inverse probability of censoring with an outcome-dependent weight,
and their Theorem~1 is a finite-sample floor minus a censoring-model term. They
predict a time and invert a $p$-value over candidate values, so the hypothesised
object determines the test weight. Shin et al.\ (2025) add a covariate-shift ratio
to this weight for an event-time interval, without a coverage theorem for the
combined weight. Cand\`es, Lei and Ren (2023) require Type-I censoring with the
censoring time observed on every unit, and the PAC-type guarantees of Gui, Hore,
Ren and Barber (2024) do not apply under loss to follow-up. Davidov et al.\ (2025)
extend double robustness to general right censoring with a finite-sample bound
whose slack in the weight-estimation error is of the same kind as ours; their
bound is a lower bound on a single survival time, holds with probability
$1-\delta$ over the calibration data rather than marginally, and does not treat a
cause. Sesia and Svetnik (2025a) handle general right censoring by imputation, and
Farina et al.\ (2026), Qin et al.\ (2025) and Si and Qiu (2025) give asymptotic
guarantees. Holmes and Marandon (2025) give a finite-sample two-sided interval on
the subjects that a trained classifier assigns to the uncensored population.
Alberge et al.\ (2025), Wang et al.\ (2026) and Dang et al.\ (2026) address the
cause or the fold restriction with, respectively, a point prediction, an
asymptotic interval, and a finite-sample statement that assumes the
exchangeability at issue. Sesia and Svetnik (2025b, 2026) form conformal
$p$-values from inverse-probability-of-censoring weights on uncensored subjects,
for survival thresholds and for selecting patients event-free past a horizon,
with asymptotic false discovery rate control and a finite-sample version whose
slack is of order $n^{-1/2}$ (2025b, Theorem~5).

For missing and coarsened outcomes, Paul et al.\ (2025) calibrate a
prediction set for a real outcome from monotonely coarsened data through the
efficient influence function, with finite-sample validity when some nuisance
functions are known. Their threshold takes an infimum over the data consistent
with the test covariates, which is the nearest precedent for the atom of
\eqref{eq:atom}. Yin, Shi, Wang and Blei (2024) and Jin, Ren and Cand\`es (2023) give the test
slot the upper bound of a response-dependent likelihood ratio, and Yi, Zhang, Tang and Wang (in press)
keep a finite-sample floor under nonignorable missingness with a
response-dependent propensity. In both, a hypothesised response determines its
weight. Fan et al.\ (2025) keep finite-sample coverage under missing covariates
with a response-dependent likelihood ratio, but the response is always observed.
A conservative finite-sample guarantee for a set of causes follows by set
inclusion from the partial-label construction of Javanmardi et al.\ (2023,
Theorem~1), which calibrates on each subject's largest score over the labels
consistent with its data; a subject censored before the horizon is consistent
with every label and enters calibration at its largest score. Liu, de Paula and
Tamer (2025) follow a related conservative route for interval-valued outcomes.
The guarantee here calibrates on the label-complete subjects and pays only for
the censoring-model error.

\section{Problem and procedures}
\label{sec:setup}

Fix a horizon $\tstar>0$. A subject has covariates $X\in\Xcal$, a latent first
event time $T>0$ and a cause $\varepsilon\in\{1,\dots,K\}$. The quantity to be
predicted is the \emph{horizon label}
\[
Y \;=\;
\begin{cases}
\varepsilon, & T\le \tstar,\\
0, & T>\tstar,
\end{cases}
\qquad Y\in\{0,1,\dots,K\},
\]
which records the cause of the first event by $\tstar$, or event-free status.
Right censoring at time $C$ gives the observed follow-up $\widetilde T=\min(T,C)$.
The label is determined, and we call the subject \emph{label-complete}, exactly
on the event
\[
\Delta \;\coloneqq\; \1\{\, C \ge m \,\}, \qquad m \;\coloneqq\; \min(T,\tstar),
\]
since a subject followed to $m$ either has had the event, with its cause
recorded, or is known to be event-free at $\tstar$.

Let $\pi_\ell(x)\in[0,1]$, $\ell\in\{0,\dots,K\}$, be fitted class probabilities
at $\tstar$, obtained from cause-specific cumulative incidence increments and the
event-free probability of a model trained on an independent split
$\mathcal D_{\mathrm{tr}}$. The nonconformity score is
\[
s(x,\ell)\;=\;1-\pi_\ell(x), \qquad V\;=\;s(X,Y)\;=\;1-\pi_Y(X),
\]
and for a threshold $q$ the prediction set is
$S_q(x)=\{\ell: \pi_\ell(x)\ge 1-q\}$, so that
\begin{equation}\label{eq:cov-as-score}
\{Y\in S_q(X)\}\;=\;\{V\le q\}.
\end{equation}
The score depends on the subject only through $(X,Y)$.

Write $G(u\mid x)=\Prob(C\ge u\mid X=x)$ for the
conditional censoring survival function, taken throughout with a closed
inequality. A marginal $G(u)$ is the special case of an $X$-free $G$. The oracle
inverse-probability-of-censoring weight and its plug-in are
\[
\wstar \;=\; \frac{1}{G(m\mid X)},
\qquad
\what \;=\; \frac{1}{\Ghat(m\mid X)},
\]
with $\Ghat$ fitted on $\mathcal D_{\mathrm{tr}}$ and therefore independent of the
calibration and test subjects. In the implementation $\Ghat$ is floored at
$g_{\min}$ before inversion, so $\what\le 1/g_{\min}$ by construction and the
conditioning event of Theorem~\ref{thm:main} holds. Positivity bounds the oracle
weights, while the floor bounds the estimated ones. Clipping an unbounded weight
is a different problem (Wang and Goel, 2026).

From a calibration sample of i.i.d.\ subjects, keep the
label-complete ones, $\Ical=\{i:\Delta_i=1\}$, with scores $V_i$ and weights
$\what_i$. Define a self-normalised weighted quantile carrying an atom of mass $A$ at
$+\infty$. With $p_i=\what_i/(\sum_{j\in\Ical}\what_j+A)$,
\begin{equation}\label{eq:qhat}
\qhat \;=\; \inf\Big\{ q:\ \textstyle\sum_{i\in\Ical} p_i\,\1\{V_i\le q\}\ \ge\ 1-\alpha \Big\},
\end{equation}
and output $S_{\qhat}(X)$ for a fresh subject.

Two procedures are compared throughout. The \emph{weighted} procedure runs
\eqref{eq:qhat} with the weights $\what_i$ and the atom of \eqref{eq:atom} below.
The \emph{unweighted}, or complete-case, procedure sets every weight and the atom
to $1$, which is split conformal prediction on the label-complete subjects. It
carries no guarantee for the full population under this selection, since
$\Delta=\1\{C\ge\min(T,\tstar)\}$ depends on the outcome through $T$.

The atom is where predicting a cause departs from weighted conformal prediction.
Lemma~\ref{lem:we} puts the test subject's own weight on the test slot, and that
weight depends on the unobserved $m=\min(T,\tstar)$, so it cannot be computed at
prediction time. The procedure uses a dominating substitute
(Lemma~\ref{lem:atomhat}),
\begin{equation}\label{eq:atom}
A_{\mathrm{env}}\;=\;\frac{1}{\max\{\Ghat(\tstar\mid X),\,g_{\min}\}}
\;\ge\;\what_{\mathrm{te}},
\end{equation}
which uses no test outcome. Atom mass can only raise the quantile, so any
dominating substitute preserves the coverage floor.

\medskip

\begin{assumption}[Independent censoring]\label{as:indep}
$C \perp\!\!\!\perp (T,\varepsilon)\mid X$. The marginal-$G$ version assumes the
stronger $C\perp\!\!\!\perp(T,\varepsilon,X)$, under which $G(\cdot\mid X)\equiv G(\cdot)$.
\end{assumption}

\begin{assumption}[Positivity]\label{as:pos}
There is $g_{\min}>0$ with $G(\tstar\mid X)\ge g_{\min}$ almost surely, hence
$\wstar\le 1/g_{\min}$.
\end{assumption}

\begin{assumption}[Sample splitting]\label{as:split}
$\pi$ and $\Ghat$ are fixed given $\mathcal D_{\mathrm{tr}}$, and the calibration
subjects and the test subject are i.i.d.\ draws independent of
$\mathcal D_{\mathrm{tr}}$.
\end{assumption}

\section{A coverage floor for the weighted procedure}
\label{sec:ceiling}
\noindent Under Assumptions~\ref{as:indep} and~\ref{as:pos},
$\E[\Delta\,\wstar\,h(X,Y)]=\E[h(X,Y)]$ for every integrable $h$
(Section~\ref{sec:suppproofs}), so by \eqref{eq:cov-as-score} population coverage
is a $\wstar$-weighted coverage over label-complete subjects. Conformal
prediction needs the corresponding statement about ranks.

\begin{lemma}[Weighted exchangeability under censoring-induced selection]\label{lem:we}
Under Assumptions~\ref{as:indep} to~\ref{as:split}, conditionally on
$|\Ical|=k$ and on the unordered
multiset of the $k$ label-complete calibration tuples together with the test
tuple, the $k+1$ tuples are weighted exchangeable in the sense of Tibshirani et
al.\ (2019, Definition~1) with weight functions $(1,\dots,1,\wstar)$, unit weight on each
label-complete calibration slot and the oracle weight on the test slot alone. Consequently, running \eqref{eq:qhat} with oracle calibration weights and
the test point's own atom $\wstar_{\mathrm{te}}$,
\[
\Prob\big(Y\in S_{\qhat}(X)\big)\;=\;\Prob\big(V_{n+1}\le\qhat\big)\;\ge\;1-\alpha .
\]
\end{lemma}

\begin{proof} In Section~\ref{sec:suppproofs}. \end{proof}

With the true $G$ in the calibration weights the floor is $1-\alpha$
(Lemma~\ref{lem:oracle}). The theorem replaces $G$ by $\Ghat$.

\label{sec:main}
\begin{theorem}[Coverage floor under estimated weights]\label{thm:main}
Run \eqref{eq:qhat} with estimated calibration weights $\what=1/\Ghat$ and the
atom $A_{\mathrm{env}}$ of \eqref{eq:atom}, where $\Ghat$ is the estimate after the
floor at $g_{\min}$ and $\Ghat(\cdot\mid x)$ is non-increasing for every $x$, so that
$A_{\mathrm{env}}\ge\what_{\mathrm{te}}$ (Lemma~\ref{lem:atomhat}). Under Assumptions~\ref{as:indep}
to~\ref{as:split}, conditionally on $\mathcal D_{\mathrm{tr}}$ and on the event
$\{\Ghat(\cdot\mid\cdot)\ge g_{\min}\}$,
\[
\Prob\big(Y\in S_{\qhat}(X)\big)\;\ge\; 1-\alpha-\Delta_w,
\qquad
\Delta_w \;\coloneqq\; \frac{1}{2\kappa}\,\E\Big|\frac{G}{\Ghat}-\kappa\Big|,
\qquad
\kappa\;\coloneqq\;\E\Big[\frac{G}{\Ghat}\Big],
\]
where $G=G(m\mid X)$, $\Ghat=\Ghat(m\mid X)$ and both expectations are over the
population law of $(X,m)$ rather than of $X$ alone. Moreover, since $\Ghat\le1$
forces $\kappa\ge\E[G]\ge g_{\min}$,
\[
\Delta_w \;\le\; \frac{1}{g_{\min}^{2}}\;
\E\!\Big[\sup_{u\le\tstar}\big|\Ghat(u\mid X)-G(u\mid X)\big|\Big].
\]
\end{theorem}

\begin{proof} In Section~\ref{sec:suppproofs}. \end{proof}

\begin{corollary}[Rate]\label{cor:rate}
Assume the marginal-$G$ version of Assumption~\ref{as:indep}, that is
$C\perp\!\!\!\perp(T,\varepsilon,X)$ jointly, and let $\Ghat$ be the marginal
Kaplan--Meier estimator of the censoring distribution, with failures removed from
the risk set before censorings at a shared time, as $T\le C$ requires. If $\tstar$ is strictly
below the end of follow-up, so that $G(\tstar)\ge g_{\min}>0$, and
$\Prob(T>\tstar)>0$, then
\[
\Prob\big(Y\in S_{\qhat}(X)\,\big|\,\mathcal D_{\mathrm{tr}}\big)
\;\ge\;1-\alpha-\Delta_w,
\qquad
\Delta_w=O_P\big(|\mathcal D_{\mathrm{tr}}|^{-1/2}\big).
\]
As $\Ghat\to G$ the penalty vanishes and the floor recovers $1-\alpha$.
\end{corollary}

\begin{proof} In Section~\ref{sec:suppproofs}. \end{proof}

\noindent The theorem holds for any atom that dominates the weight the procedure
carries, and \eqref{eq:atom} is the smaller of the two atoms that
Lemma~\ref{lem:atomhat} admits. Section~\ref{sec:suppatom} gives that lemma, the
condition on calibration mass under which the weighted quantile is finite at a
small $\alpha$, and the per-cause version of the procedure.

\section{Experiments}
\label{sec:simulation}

Covariates are independent standard normals, cause-specific hazards are constant
with cause~1 depending on the first covariate, the horizon is $\tstar=1$ and the
all-cause survival there is $0.5$ at zero signal (the share parameter). Censoring is exponential and, unless a design
states otherwise, covariate-free, so coverage is computed against the generating
process's own censoring survival and the evaluation shares no estimate with the
procedure. The designs use $g_{\min}=0.05$, between $2.5$ and $17$ times looser
than Assumption~\ref{as:pos} permits in them; designs with another value state it
in their captions. Section~\ref{sec:suppsim} gives the remaining constants.

Section~\ref{sec:correctfit} holds the fitted model correct,
Section~\ref{sec:dials} misspecifies it in a controlled direction, and
Section~\ref{sec:wrongG} keeps the score exact and misspecifies the censoring
model, covering three cells of Table~\ref{tab:twobytwo}.

\subsection{A correctly specified score}
\label{sec:correctfit}

\begin{figure}[t]
\centering
\includegraphics[width=0.82\textwidth]{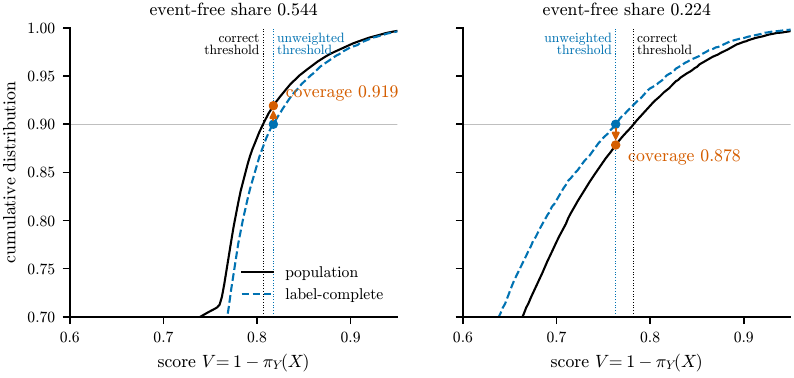}
\caption{Where selection puts the scores, at the true score and censoring $0.5$. Each panel shows the population distribution of the score (solid) and its distribution among label-complete subjects (dashed), pooled over $20$ replications with the seeds and design of Table~\ref{tab:oraclescore}. The unweighted threshold is where the dashed curve reaches $1-\alpha$ (blue dot), and coverage is the height of the solid curve at that threshold (orange dot); the correct threshold is where the solid curve reaches $1-\alpha$. At a large event-free share selection moves the dashed curve to the right of the solid one near $1-\alpha$, so the unweighted threshold is too large and coverage is $0.919$. At a small share it moves it left, and coverage is $0.878$. The curves and coverages are recomputed from the scores without the forest (\texttt{e57\_score\_cdf.py}), so they differ slightly from Table~\ref{tab:oraclescore}.}
\label{fig:mechanism}
\end{figure}

\begin{table}[t]
\centering
\caption{The searches of Section~\ref{sec:correctfit} at the true score. Each replication fits one forest and swaps only the score, so the last column is paired. $\pi_0$ is the mean event-free probability at $t^\ast$, true and as fitted. The first block sweeps the event-free share at censoring $0.50$, the second the censoring rate at the largest share. Coverage is inverse-probability weighted under the oracle $G$ over 20 replications, with standard errors in units of the last printed digit. The weighted arm uses training-fold Kaplan--Meier weights and the envelope atom \eqref{eq:atom}. Each draw normalises the cause-1 hazard by its own covariate mean, so the calibration pool and the test draw do not share one law exactly, and this table is a paired comparison of scores, not i.i.d.\ evidence. Table~\ref{tab:iid} reruns the true-score cells with independent draws.}
\label{tab:oraclescore}
\resizebox{\ifdim\width>\textwidth\textwidth\else\width\fi}{!}{%
\begin{tabular}{rrrrrrr}
\toprule
censoring & true $\pi_0$ & fitted $\pi_0$ & unwtd., fitted & unwtd., true & wtd., true & fitted $-$ true \\
\midrule
0.50 & 0.544 & 0.571 & 0.9225 & 0.9194(25) & 0.9030(27) & +0.0031(24) \\
0.50 & 0.322 & 0.425 & 0.9193 & 0.9031(31) & 0.9021(29) & +0.0161(34) \\
0.50 & 0.224 & 0.385 & 0.9159 & 0.8774(35) & 0.9096(35) & +0.0385(30) \\
\addlinespace
0.20 & 0.544 & 0.544 & 0.9092 & 0.9056(23) & 0.9006(25) & +0.0036(19) \\
0.40 & 0.544 & 0.556 & 0.9142 & 0.9138(27) & 0.9022(28) & +0.0004(20) \\
0.60 & 0.544 & 0.587 & 0.9277 & 0.9254(26) & 0.9043(30) & +0.0022(26) \\
0.70 & 0.544 & 0.616 & 0.9366 & 0.9344(30) & 0.9056(36) & +0.0022(27) \\
0.75 & 0.544 & 0.637 & 0.9391 & 0.9372(35) & 0.9050(55) & +0.0019(38) \\
\bottomrule
\end{tabular}
}
\end{table}

Table~\ref{tab:oraclescore} runs the unweighted procedure at the true class
probabilities. The share parameter $0.5$, $0.25$ and $0.15$ gives realised
event-free shares of $0.544$, $0.322$ and $0.224$. At the largest share the
unweighted procedure over-covers, increasingly with censoring, from $0.9056$ at
censoring $0.20$ to $0.9372$ at $0.75$. At the smallest it under-covers, at
$0.8774$. The weighted procedure stays between $0.9006$ and $0.9096$. That table
normalises each draw by its own covariate mean and averages $20$ replications, so
the calibration pool and the test draw do not share one law exactly. With a fixed
normaliser, so that all subjects are i.i.d., and $200$ replications
(Table~\ref{tab:iid}), the smallest share covers $0.8723$, $21$ standard errors
below nominal, the middle share $0.9008$, within one standard error of nominal,
and the largest share $0.9053$ to $0.9380$ across censoring rates, while the
weighted procedure stays between $0.9005$ and $0.9054$.

Figure~\ref{fig:mechanism} shows the mechanism. The unweighted threshold is a
quantile of the label-complete scores, and selection shifts their distribution
relative to the population's near $1-\alpha$. Subjects event-free at the horizon
have the lowest selection probability at fixed covariates and are
under-represented in complete-case
calibration. At a large share the true score ranks their label as the easiest,
and under-representing them raises the threshold. At a small share it ranks their
label among the hardest, and under-representing them lowers it.

A fitted forest overstates the event-free probability, by $0.027$ at the largest
share and $0.160$ at the smallest, which pushes the unweighted procedure up. At a
fitted score the unweighted procedure therefore does not fall below nominal
across censoring rates, event-free shares and miscoverage levels
(Figure~\ref{fig:alpha} and \supmat). At
the true score the censoring and miscoverage-level sweeps still over-cover, with
the smallest margin $0.9549$ against $0.950$ (\texttt{e53-cells.json}), and the
event-free-share sweep under-covers as above.

At a fitted score, Figure~\ref{fig:alpha} crosses the miscoverage level with the censoring rate over
16 cells, and the weighted procedure returns sets no larger than the unweighted
ones in all 16. The arms share seeds, forest and splits, so the differences are
paired, and all 16 are negative. Relative to the unweighted set the reduction
ranges from $0.52\%$ at $\alpha=0.05$ and censoring $0.2$, where it is
$-0.0134\pm0.0020$ labels, to $25.5\%$ at $\alpha=0.3$ and censoring $0.75$.
\Supmat\ compares both procedures with a set read directly
off the fitted cumulative incidence functions.

\subsection{A misspecified score}
\label{sec:dials}

To misspecify the fitted model in a known direction we shrink its event-free
probability by a factor $1-b$ and renormalise. The shrinkage is applied to the
forest's output, so no forest is refitted and the misspecification is separate
from the covariate noise that Table~\ref{tab:scoremisspec} varies. The origin
$b=0$ is the forest, not the truth. In the exponential family the forest's mean
event-free probability exceeds the true one by up to $0.072$ at censoring $0.70$
(Table~\ref{tab:oraclescore}), so the crossings below are measured from a slight
overstatement. Coverage is weighted by the oracle $G$, so an error in $\hat G$
cannot offset itself.

\begin{figure}[t]
\centering
\includegraphics[width=0.92\textwidth]{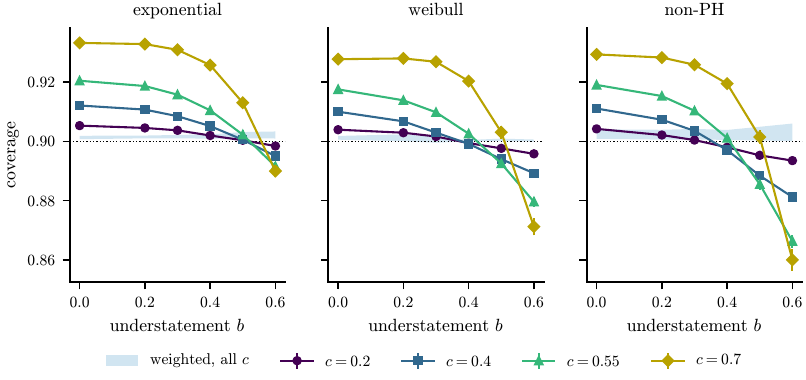}
\caption{Coverage of the unweighted procedure against the understatement $b$ of the event-free probability, at $250$ replications (Table~\ref{tab:dialsdense}), with bars of two marginal standard errors. Each line is a censoring rate $c$, and the dotted line is nominal. The lines cross nominal at different $b$, so no censoring rate marks where the procedure becomes unsafe. The shaded band spans the weighted procedure over all four censoring rates.}
\label{fig:dials}
\end{figure}

\begin{table}[t]
\centering
\caption{The censoring rate crossed with the direction of the fitted model's error. $b$ shrinks the fitted event-free probability and renormalises, and the covariate is clean, so $b$ is the only misspecification. Coverage is inverse-probability weighted under the oracle $G$. The coverage columns carry the marginal standard error over replications and the swing column the paired one, since the $b$ sweep reuses one forest fit per cell, both in units of the last printed digit. Table~\ref{tab:dialsdense} reruns this design at 250 replications on a finer $b$ grid.}
\label{tab:dialsbody}
\resizebox{\ifdim\width>\textwidth\textwidth\else\width\fi}{!}{%
\begin{tabular}{lrrrrrrr}
\toprule
family & $c$ & \multicolumn{4}{c}{unweighted} & \multicolumn{2}{c}{weighted} \\
 & & $b=0$ & $b=0.4$ & $b=0.6$ & swing & $b=0$ & $b=0.6$ \\
\midrule
exponential & 0.20 & 0.9039(14) & 0.9008(16) & 0.8979(17) & -0.0061(17) & 0.8997(15) & 0.9005(17) \\
  & 0.40 & 0.9132(14) & 0.9055(19) & 0.8956(19) & -0.0176(17) & 0.9012(12) & 0.9012(17) \\
  & 0.55 & 0.9199(13) & 0.9121(18) & 0.8901(22) & -0.0298(18) & 0.8991(15) & 0.9004(19) \\
  & 0.70 & 0.9318(15) & 0.9259(19) & 0.8896(31) & -0.0422(31) & 0.9004(24) & 0.9019(30) \\
weibull & 0.20 & 0.9020(15) & 0.8981(12) & 0.8954(14) & -0.0067(20) & 0.8987(15) & 0.8995(14) \\
  & 0.40 & 0.9085(16) & 0.8974(13) & 0.8859(18) & -0.0226(22) & 0.8987(17) & 0.8977(17) \\
  & 0.55 & 0.9157(16) & 0.8995(18) & 0.8769(21) & -0.0389(27) & 0.8971(17) & 0.8989(22) \\
  & 0.70 & 0.9263(17) & 0.9178(21) & 0.8676(35) & -0.0587(35) & 0.9000(24) & 0.8992(28) \\
non-PH & 0.20 & 0.9032(14) & 0.8975(14) & 0.8951(14) & -0.0082(17) & 0.8998(15) & 0.9014(13) \\
  & 0.40 & 0.9099(16) & 0.8953(19) & 0.8836(17) & -0.0262(20) & 0.9005(18) & 0.9008(14) \\
  & 0.55 & 0.9190(14) & 0.9023(22) & 0.8696(29) & -0.0494(29) & 0.8999(15) & 0.9046(22) \\
  & 0.70 & 0.9301(17) & 0.9177(23) & 0.8614(48) & -0.0687(47) & 0.9028(21) & 0.9075(34) \\
\bottomrule
\end{tabular}
}
\end{table}

Table~\ref{tab:dialsbody} crosses $b$ with the censoring rate. At $b=0$ the
unweighted procedure's over-coverage grows with censoring, from $+0.0020$ to
$+0.0039$ at censoring $0.20$ to $+0.0263$ to $+0.0318$ at $0.70$, as in
Section~\ref{sec:correctfit} under a different generator. The coverage loss caused
by $b$ also grows with censoring, from $-0.0061$ to $-0.0082$ at $0.20$ to
$-0.0422$ to $-0.0687$ at $0.70$, a factor of $7.0$ to $8.8$ with paired standard
errors of $2.0$ and $2.7$.

Censoring thus amplifies two effects of opposite sign: the over-coverage at $b=0$
and the loss from a misspecified fit. At $b=0$ the unweighted procedure is above
nominal by at least 2 marginal standard errors in 11 of the 12 cells. At $b=0.6$
it is below nominal in every cell, by that margin in 11, so at every censoring rate
the procedure crosses nominal at some intermediate $b$, and the crossing point
varies with the censoring rate. At $250$ replications (Table~\ref{tab:dialsdense}), non-PH at $b=0.4$
covers $0.8971$ at censoring $0.40$, $4.5$ standard errors below nominal, and
$0.9195$ at $0.70$. Weibull at $b=0.5$ covers $0.8940$ and $0.9031$, $9.5$ below
and $2.9$ above. Within one family and at a fixed misspecification, raising the
censoring rate moves the procedure from below nominal to above it. The observed
crossing point moves with the censoring rate (Table~\ref{tab:contour},
Figure~\ref{fig:dials}), so no threshold on the censoring rate alone separates safe
from unsafe cases. The weighted procedure
ranges from $0.8971$ to $0.9075$ across the grid.

\paragraph{Where the error lands.} Table~\ref{tab:landing} decomposes the
unweighted procedure's miscoverage by true label. Subjects event-free at the
horizon have $m=\tstar$, so their selection probability is $G(\tstar)$, the
smallest possible, and the calibration fold is depleted of them. At censoring
$0.70$ their selection rate is $0.1960$, against $0.5604$ and $0.5332$ for the two
causes. At $b=0$ the unweighted procedure covers this group $0.9995$ of the time
and all of its error falls on the cause labels. At $b=0.6$ the group carries
$+0.0700$ of additional miscoverage while both causes improve.

Since $b$ raises the event-free group's nonconformity by construction, we compare
with a control $\lambda$ that replaces the informative covariate with noise in
training. It changes the fitted model, taking cause-2 coverage from $0.8282$ to
$0.7726$, but leaves the event-free group's miscoverage at $0.0000$, a change of
$-0.0003$ from the $b=0$ row. A misspecification that leaves the depleted group
alone does not move marginal coverage. The effect also scales with the depletion:
the same $b$ adds $+0.0165$ of miscoverage at censoring $0.40$ and $+0.0430$ at
$0.70$, while the event-free selection rate falls from $0.6331$ to $0.1960$.

\paragraph{A common misspecification.} A standard competing-risks error acts in
the same direction as $b$. Estimating each cause's incidence as one minus a
cause-specific survival curve treats the competing event as censoring and
understates the event-free probability. We ran this error in all three families
over 48 designs at 100 replications each, taking each cause's curve as one minus
the exponential of the negative of the forest's cause-specific Nelson--Aalen cumulative hazard and the
event-free probability as one minus their sum, clipped and renormalised
(Table~\ref{tab:misfitdirect}). The equivalent $b$
reaches $0.1764$ under exponential hazards, $0.2464$ under Weibull and $0.3088$
under non-proportional hazards. These designs use per-draw normalisation and, for Weibull and non-proportional
hazards, a censoring rate set from each draw (Appendix~\ref{sec:suppsim}). The
unweighted procedure falls at least two standard errors below nominal in 7
designs, all Weibull or non-proportional at
censoring $0.20$ to $0.55$, and at least four standard errors below in 4, at
censoring $0.20$ to $0.40$. The lowest coverage is $0.8920$, under
non-proportional hazards at censoring $0.40$ and competing fraction $0.6$. In a
paired comparison on the exponential family (Table~\ref{tab:naturalmisfit}) the
weighted procedure does not move with the error, at $t$ between $-1.1$ and $+0.9$
in all 8 cells, while the unweighted procedure falls significantly in 7 of them
and stays above nominal. Whether this error alone breaks unweighted calibration
depends on the hazard shape rather than on the size of the error.

\paragraph{Public cohorts.} Three competing-risks cohorts from the R package
\texttt{survival} show the same direction. Their censoring depends on the
covariates, and in these analyses coverage is estimated with the marginal
Kaplan--Meier model that the weighted procedure also calibrates on, so the
estimates do not establish population coverage. Under this estimator the
unweighted procedure as fitted is estimated above nominal in all 6 cells, by
$0.003$ to $0.029$. At $b=0.6$ the estimated coverage at the later horizon is
$0.8913$ for mgus2 and $0.8966$ for PBC. These estimates use a misspecified
censoring model, and the mgus2 cell also violates Assumption~\ref{as:pos} in the
latest quarter of accrual (Section~\ref{sec:public}).

\subsection{A misspecified censoring model}
\label{sec:wrongG}

\begin{figure}[t]
\centering
\includegraphics[width=0.92\textwidth]{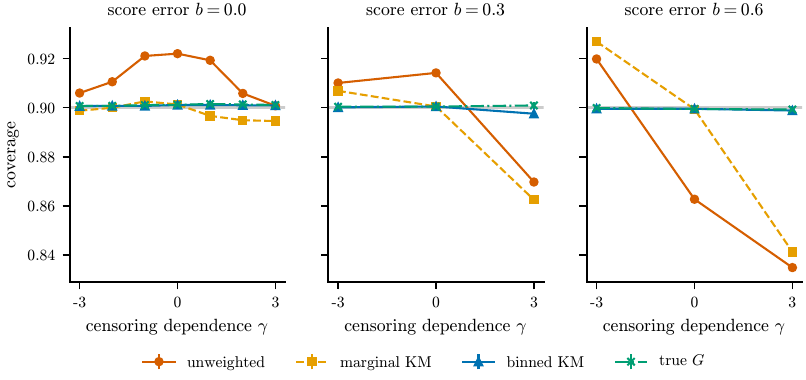}
\caption{Coverage when censoring depends on the covariate carrying the signal, on bounded support with $G(\tstar\mid X)\ge0.08$. The censoring hazard is proportional to $\exp(\gamma X_0)$, and $b$ shrinks the true event-free probability by $1-b$, so the left panel has the true score. Weights from the true $G$, or from binned Kaplan--Meier, which approximates it (cool colours), stay on nominal (grey line) in every panel. The unweighted procedure and marginal Kaplan--Meier weights move with both $\gamma$ and $b$, in either direction. Bars are two standard errors over $200$ replications. Tables~\ref{tab:boundedcens} and~\ref{tab:fourthcell} give the numbers.}
\label{fig:wrongG}
\end{figure}

The weighted procedure is robust to errors in the score only if its weights are
correct. Figure~\ref{fig:wrongG} and Table~\ref{tab:boundedcens} make censoring
depend on the covariate that carries the signal, on bounded support so that
$G(\tstar\mid X)\ge0.08$ for every subject, and use the oracle score so that the
only error is in $\Ghat$. The unweighted procedure stays above nominal, lowest at
$0.9008$. Marginal Kaplan--Meier weights are at least two standard errors below
nominal in 4 of the 6 designs with covariate-dependent censoring, at $0.8946$ to
$0.8988$, while Kaplan--Meier within 10 bins of the covariate holds at $0.9006$ or
above. The shortfalls are small and in the direction Theorem~\ref{thm:main}
allows, since a marginal $\Ghat$ is misspecified once censoring depends on the
covariates and $\Delta_w$ is then positive.

Stronger covariate dependence breaks both procedures. In a two-region law with the
true score, $40\%$ of subjects are in a region where $30\%$ are event-free at the
horizon and $G(\tstar)=0.10$, and the rest in a region where $95\%$ are event-free
and $G(\tstar)=0.90$. The unweighted procedure and marginal Kaplan--Meier weights
both cover $0.8294$ with single-label sets, while Kaplan--Meier weights fitted
within each region and the oracle weights cover $0.9496$ with $1.40$ labels
(\texttt{e54-cells.json}, 200 replications). Complete-case selection retains the
region-B survivors, whose label the true score ranks easiest, and drops the
region-A survivors, whose label it ranks among the hardest, so the complete-case
threshold is too small. Marginal weights cannot undo selection between regions.
Most subjects, $0.69$, are event-free at the horizon, and the true score still
under-covers, so the event-free share of Section~\ref{sec:correctfit} indicates
where selection moves the scores but does not determine it.

The right two panels of Figure~\ref{fig:wrongG} (Table~\ref{tab:fourthcell}) also
misspecify the score, shrinking the true event-free probability by $1-b$ in the
design of Table~\ref{tab:boundedcens}. The unweighted procedure misses in either
direction: at $b=0.6$ it covers $0.8349$ when censoring is heavier where cause~1
is likelier ($\gamma=3$) and $0.9198$ when it leans the other way. Marginal
Kaplan--Meier weights follow it, from $0.8413$ to $0.9269$. Oracle weights stay
between $0.8992$ and $0.9009$, and binned Kaplan--Meier weights within $0.0012$ of
nominal except one cell at $0.8975$ (\texttt{e55-cells.json}).

\paragraph{Real covariates.} Table~\ref{tab:realcov} keeps the covariates and
fitted covariate effects of two public cohorts, mgus2 and PBC, takes a
cause-specific exponential regression fitted to each as the truth, and adds
synthetic censoring that depends on each subject's true all-cause hazard, so the
true $G$ is known. Weights from the true $G$, or from binned Kaplan--Meier, stay
between $0.8995$ and $0.9030$ in all $12$ cells. At the true score and
$\gamma=+3$, the unweighted procedure on mgus2 covers $0.8885$ and marginal
Kaplan--Meier weights $0.8871$, while binned Kaplan--Meier covers $0.9018$, so the
upper-right cell of Table~\ref{tab:twobytwo} also appears on real covariates. On
PBC at the true score with covariate-free censoring the unweighted procedure
over-covers, at $0.9203$, and an understated event-free probability lowers it by
$0.0258$ on mgus2 and $0.0327$ on PBC, as in Section~\ref{sec:dials}. Under a
misspecified score it again misses in either direction, from $0.8611$ to $0.9206$
(\texttt{e58-cells.json}).

\paragraph{Outside positivity.} When Assumption~\ref{as:pos} fails, coverage
degrades gradually. In a design where only the population mass $p$ whose true
$G(\tstar\mid X)$ equals half of $g_{\min}$ varies, coverage against the latent
label is $0.898$, $0.896$, $0.894$ and $0.883$ as $p$ goes from $0$ to $0.30$, and
on the violating subjects themselves $0.895$, $0.876$ and $0.826$
(Table~\ref{tab:posviolation}). The clip at $g_{\min}$ cannot detect this: it acts
on the analyst's estimate, and a marginal Kaplan--Meier averages over the hard
region, so the clip fires at rate $0.0000$ at every $p$, including the one where
$30$ per cent of subjects are below the bound. Reweighting cannot repair a
violated identity.

\section{Discussion and limitations}
\label{sec:limitations}

\noindent In the studied simulation families, a correct outcome model does not
protect the unweighted procedure. At the true score it under-covers when few subjects are event-free at the horizon,
and under covariate-dependent censoring it can fail badly. A fit that understates
the event-free probability lowers coverage at any censoring rate, and whether the
total crosses nominal depends jointly on the share, the size of the error, the
censoring rate and the hazard shape. In these designs, weights from a correct
censoring model kept the weighted procedure near nominal across all of these, and
it is the only one of the two procedures with a guarantee, conditional on $\Ghat$. With a misspecified censoring model it fails with the
unweighted procedure, and marginal Kaplan--Meier weights fall below nominal in 4
of the 6 covariate-dependent designs of Table~\ref{tab:boundedcens}. In these
synthetic and semi-synthetic designs the practical question is whether the
censoring can be modelled, not whether the outcome model is good. On SEER at
$\tstar=107$ months, the one real censoring mechanism studied, correcting the
censoring model changed pooled coverage by $0.0011$, and the marginal model was
closer to nominal (Section~\ref{sec:realdata}).

\noindent The cells of Table~\ref{tab:twobytwo} come from different designs.
Section~\ref{sec:dials} uses a fitted forest and covariate-free censoring, and
Section~\ref{sec:wrongG} the true score and covariate-dependent censoring, so their
magnitudes are not comparable. The synthetic designs outside
Section~\ref{sec:wrongG} share one structure: independent Gaussian covariates, a
cause-1 signal on the first covariate and covariate-free exponential censoring,
with three hazard families. The real-covariate design of Table~\ref{tab:realcov}
departs from it, but its event times are exponential and its censoring synthetic.

\noindent The bound on $\Delta_w$ is a total-variation coupling with a
$g_{\min}^{-2}$ constant. It exceeds $\alpha=0.1$ in every design we run, while
the measured $\Delta_w$ ranges from $0.0035$ to $0.2371$, so it establishes the
parametric rate in $\lVert\Ghat-G\rVert$ but gives no usable number. The penalty
is a population quantity in the unknown $G$ and cannot be estimated from the data.

\noindent The shrinkage $b$ is applied to the fitted model's output rather than
obtained by fitting a different model, and at $b=0.6$ it understates the
event-free rate by a factor of $2.5$ in the cell that crosses. The competing-risks
error of Section~\ref{sec:dials} reaches the crossing without shrinkage, under two
of the three hazard families. Fitted-score experiments use a competing-risks forest,
except Table~\ref{tab:estimator}, and oracle-score experiments use the known
class probabilities. Standard errors on the public cohorts
are over splits of one fixed sample, so those rows describe the procedure on those
cohorts.

\label{end:maintext}

\clearpage
\section*{Reproducibility statement}

The code that produced every table and figure accompanies the submission, with
the fixed seeds of each experiment. The tables report runs at the atom
\eqref{eq:atom}, except two tables in \supmat, the corner design, which uses
oracle weights and the constant atom $1/g_{\min}$, and the per-class thresholds
of the independence-violation design, which use the constant atom. Assumptions~\ref{as:indep}
to~\ref{as:split} are stated in Section~\ref{sec:setup}, and every proof is in
Section~\ref{sec:suppproofs}. The simulation designs, their constants and the
number of replications behind each exhibit are in Section~\ref{sec:suppsim}. Table
entries are generated from the run outputs rather than transcribed.

SEER data are released under a data-use agreement that does not permit
redistribution, so the registry cohort cannot accompany this paper. We provide the
selection statement that defines it, the two filters and their sensitivity
analyses (Section~\ref{sec:realdata}), and the aggregate outputs behind every
registry number in the text. A reader holding a SEER release can rebuild the
cohort from the stated selection.

\section*{Ethics statement}

The registry analysis uses the SEER public-use research data, which are
de-identified at source and obtained under the standard data-use agreement. No
individual-level record is reported here and no subject is identifiable from any
aggregate in this paper. The work is a methodological study of uncertainty
quantification and was not used to guide the care of any patient.

\section*{Generative AI disclosure}

An AI coding assistant (Anthropic's Claude, used through Claude Code) assisted
with implementing the experiments and with preparing this manuscript. The
research questions, the methodology and the final content of every claim were
decided by the authors, who reviewed all AI-assisted output and take full
responsibility for the correctness of the code and the accuracy of the claims.

\section*{References}
\begin{sloppypar}
\begin{list}{}{%
  \setlength{\leftmargin}{1.2em}\setlength{\itemindent}{-1.2em}%
  \setlength{\itemsep}{0pt}\setlength{\parsep}{0pt}\setlength{\topsep}{0pt}}
\item J. Alberge, V. Maladi\`ere, O. Grisel, J. Ab\'ecassis and G. Varoquaux.
  Survival models: proper scoring rule and stochastic optimization with competing
  risks. In \emph{Proceedings of the Twenty-Eighth International Conference on
  Artificial Intelligence and Statistics}, PMLR 258, 2025.
\item H. Bostr\"om, U. Johansson and A. Vesterberg. Predicting with confidence
  from survival data. In \emph{Proceedings of the Eighth Symposium on Conformal and
  Probabilistic Prediction and Applications}, PMLR 105:123--141, 2019.
\item H. Bostr\"om, H. Linusson and A. Vesterberg. Mondrian predictive systems
  for censored data. In \emph{Proceedings of the Twelfth Symposium on Conformal
  and Probabilistic Prediction with Applications}, PMLR 204:399--412, 2023.
\item E. J. Cand\`es, L. Lei and Z. Ren. Conformalized survival analysis.
  \emph{Journal of the Royal Statistical Society Series B}, 85(1):24--45, 2023.
\item Q.-V. Dang, N.-S.-A. Nguyen and T.-B.-D. Vo. CONFIDE: conformal free
  inference for distribution-free estimation in causal competing risks.
  \emph{Mathematics}, 14(2):383, 2026.
\item H. Davidov, S. Feldman, G. Shmai, R. Kimmel and Y. Romano. Conformalized
  survival analysis for general right-censored data. In \emph{International
  Conference on Learning Representations}, 2025.
\item J. Fan, J. Park, T. P. T. Vo and N. Brunel. Weighted conformal prediction
  provides adaptive and valid mask-conditional coverage for general missing data
  mechanisms. arXiv:2512.14221, 2025.
\item R. Farina, E. J. Tchetgen Tchetgen and A. K. Kuchibhotla. Doubly robust and
  efficient calibration of prediction sets for right-censored time-to-event
  outcomes. \emph{Biometrika}, 113(3):asag036, 2026.
\item R. D. Gill. \emph{Censoring and Stochastic Integrals}. Mathematical Centre
  Tracts 124. Mathematisch Centrum, Amsterdam, 1980.
\item Y. Gui, R. Hore, Z. Ren and R. Foygel Barber. Conformalized survival analysis
  with adaptive cut-offs. \emph{Biometrika}, 111(2):459--477, 2024.
\item C. Holmes and A. Marandon. Two-sided conformalized survival analysis.
  arXiv:2410.24136, 2025.
\item A. Javanmardi, Y. Sale, P. Hofman and E. H\"ullermeier. Conformal
  prediction with partially labeled data. In \emph{Proceedings of the Twelfth
  Symposium on Conformal and Probabilistic Prediction with Applications}, PMLR
  204:251--266, 2023.
\item Y. Jin, Z. Ren and E. J. Cand\`es. Sensitivity analysis of individual
  treatment effects: a robust conformal inference approach. \emph{Proceedings of
  the National Academy of Sciences}, 120(6):e2214889120, 2023.
\item L. Lei and E. J. Cand\`es. Conformal inference of counterfactuals and
  individual treatment effects. \emph{Journal of the Royal Statistical Society
  Series B}, 83(5):911--938, 2021.
\item W. Liu, \'A. de Paula and E. Tamer. Prediction sets and conformal
  inference with interval outcomes. arXiv:2501.10117, 2025.
\item M. Paul, A. K. Kuchibhotla and E. J. Tchetgen Tchetgen. Multiply robust
  conformal risk control with coarsened data. arXiv:2508.15489, 2025.
\item A. Podkopaev and A. Ramdas. Distribution-free uncertainty quantification
  for classification under label shift. In \emph{Proceedings of the
  Thirty-Seventh Conference on Uncertainty in Artificial Intelligence}, PMLR
  161:844--853, 2021.
\item J. Qin, J. Piao, J. Ning and Y. Shen. Conformal predictive intervals in
  survival analysis: a resampling approach. \emph{Biometrics}, 81(2):ujaf063, 2025.
\item M. Sesia and V. Svetnik. Doubly robust conformalized survival analysis with
  right-censored data. arXiv:2412.09729, 2025a.
\item M. Sesia and V. Svetnik. Conformal survival bands for risk screening under
  right-censoring. In \emph{Proceedings of the Fourteenth Symposium on Conformal and
  Probabilistic Prediction with Applications}, PMLR 266:464--514, 2025b.
\item M. Sesia and V. Svetnik. Distribution-free selection of low-risk oncology
  patients for survival beyond a time horizon. arXiv:2512.18118, 2026.
\item J. Shin, C. H. Lee and S. Kang. Weighted conformal prediction for
  survival analysis under covariate shift. arXiv:2512.03738, 2025.
\item W. Si and H. Qiu. Training-set conditionally valid prediction sets with
  right-censored data. arXiv:2505.02213, 2025.
\item R. J. Tibshirani, R. Foygel Barber, E. J. Cand\`es and A. Ramdas. Conformal
  prediction under covariate shift. \emph{Advances in Neural Information
  Processing Systems} 32, 2019.
\item V. Vovk, D. Lindsay, I. Nouretdinov and A. Gammerman. Mondrian confidence
  machine. On-line Compression Modelling Project, Working Paper 4, 28 March 2003.
\item J. Wang and S. Goel. Weight clipping for robust conformal inference under
  unbounded covariate shifts. arXiv:2605.02072, 2026.
\item Y. Wang, A. W. Levis, S. Yang and L. Han. History-aware conformal prediction
  sets for censored time-to-event outcomes. arXiv:2605.06581, 2026.
\item M. Yi, Z. Xiao, H. J. Wang and Y. Tang. Survival conformal prediction under
  random censoring. \emph{Stat}, 14(2):e70052, 2025.
\item M. Yi, Y. Zhang, Y. Tang and H. J. Wang. Conformal prediction under
  nonignorable missingness. \emph{Statistica Sinica}, in press,
  doi:10.5705/ss.202025.0156.
\item M. Yin, C. Shi, Y. Wang and D. M. Blei. Conformal sensitivity analysis for
  individual treatment effects. \emph{Journal of the American Statistical
  Association}, 119(545):122--135, 2024.
\end{list}
\end{sloppypar}

\clearpage
\appendix

\renewcommand{\theremark}{A\arabic{remark}}
\setcounter{remark}{0}

\section{Oracle weights, the atom and per-cause thresholds}
\label{sec:suppatom}

This section collects the parts of the theory the experiments do not turn on.
The first is the floor at oracle weights, which Theorem~\ref{thm:main} extends
to estimated ones.

\begin{lemma}[Oracle weights]\label{lem:oracle}
Run \eqref{eq:qhat} with oracle calibration weights $\wstar_i$ and any atom $A$
that dominates the test slot's own weight, $A\ge\wstar_{\mathrm{te}}$ almost
surely. Under Assumptions~\ref{as:indep} to~\ref{as:split},
\[
\Prob\big(Y\in S_{\qhat}(X)\big)\;\ge\; 1-\alpha .
\]
The constant $w_\infty=1/g_{\min}$ is such an atom by Assumption~\ref{as:pos},
and so is the smaller $1/G(\tstar\mid X)$ (Remark~\ref{rem:atomtight}).
\end{lemma}

\begin{proof} In Section~\ref{sec:suppproofs}. \end{proof}

\begin{lemma}[Which weight the atom must dominate]\label{lem:atomhat}
Theorem~\ref{thm:main} and Corollary~\ref{cor:mondrian} hold for any atom $A$ with
$A\ge\what_{\mathrm{te}}$ almost surely, where $\what_{\mathrm{te}}$ is the
calibration weight function the procedure itself uses, evaluated at the test
subject. This is one rule and not two, Lemma~\ref{lem:oracle} being its case
$\what=\wstar$. Two atoms qualify, both computable and neither adding a
condition to the theorem.
\begin{enumerate}
\item[\rm(i)] The envelope $A_{\mathrm{env}}$ of \eqref{eq:atom}, whenever
$\Ghat(\cdot\mid X)$ is non-increasing in its time argument, which is a property
of the estimator, required here rather than assumed of $G$, and holding whenever
$\Ghat$ is one minus a non-decreasing cumulative estimate, as both estimators run
in this paper are. Given it, $m\le\tstar$ gives
$\max\{\Ghat(\tstar\mid X),g_{\min}\}\le\max\{\Ghat(m\mid X),g_{\min}\}
=1/\what_{\mathrm{te}}$. The same $\Ghat$ stands on both sides, which is what
makes this an identity rather than a statistical claim. An envelope built from a
different estimator than the one supplying the weights carries no such
guarantee, and the second and third rows of Table~\ref{tab:atomtarget} measure
what that costs.
\item[\rm(ii)] The constant $w_\infty=1/g_{\min}$, because the implementation
clips $\Ghat$ at $g_{\min}$ before inversion, so $\what\le1/g_{\min}$ identically
and no event is involved here either. The clip satisfies the conditioning event
of Theorem~\ref{thm:main} by construction.
\end{enumerate}
Since $A_{\mathrm{env}}\le w_\infty$ always, (i) is the better of the two, and
(ii) is the one admissible under the oracle weights as well.
\end{lemma}

\begin{proof} In Section~\ref{sec:suppproofs}. \end{proof}

\begin{table}[t]
\centering
\caption{Which weight the atom must dominate. Domination is checked on every label-complete test subject at censoring strength $c_s\in\{0,0.5,1,2,3\}$, covariate-free at $c_s=0$, for three atoms against two targets, the weight $\what_{\mathrm{te}}$ the procedure carries and the oracle $\wstar_{\mathrm{te}}$. Entries are violation rates. The shipped envelope is built from the same $\Ghat$ that supplies the weights. The other two atoms use the true $G$, so a violation there comes from the mismatch between estimators and not from estimation error.}
\label{tab:atomtarget}
\begin{tabular}{lrr}
\toprule
atom & violations of $\what_{\mathrm{te}}$ & violations of $\wstar_{\mathrm{te}}$ \\
\midrule
$A_{\mathrm{env}}$ from $\Ghat$, the shipped atom & 0 & 52,087 (18.1\%) \\
$A_{\mathrm{env}}$ from a different estimator, here the true $G$ & 147,927 (51.5\%) & 0 \\
$1/G(m\mid X)$, likewise not the weights' estimator & 196,927 (68.5\%) & 0 \\
\bottomrule
\end{tabular}
\end{table}

\noindent Which target is meant decides whether the proof applies, and
Table~\ref{tab:atomtarget} measures it. Domination is checked on every
label-complete test subject under covariate-dependent censoring, for three atoms
against both candidate targets. Every construction changes verdict with the target. The shipped envelope never
violates the weight the procedure carries but does violate the oracle weight, and
the two atoms whose estimator is not the one supplying the weights are exactly
the other way round. It is the first column that Theorem~\ref{thm:main}
requires. Those two rows are the failure the lemma rules out. They use the true
$G$, which carries no estimation error, so what breaks domination is the mismatch
and not a deficiency of the atom's own estimate. An analyst whose weights come
from a conditional estimate while the envelope is read off a marginal one stands
in the same relation and holds no oracle.

\begin{remark}[An operating condition at small $\alpha$]\label{rem:alpha}
The atom is not free at small $\alpha$. The weighted quantile targets
$(1-\alpha)(W+A)$ with $W=\sum_{i\in\Ical}\what_i$ and $A$ the atom, so the
threshold is $+\infty$, and the prediction set is the whole label space, exactly
when
\[
W \;<\; A\,\frac{1-\alpha}{\alpha}.
\]
This is the classical split-conformal degeneracy $n<1/\alpha-1$ with the atom's
mass in place of the sample size. Censoring enters only through $A$.
At \eqref{eq:atom} the boundary is small and set by the data, not by the analyst.
Section~\ref{sec:suppreading} gives the check to run before forming a set.
\end{remark}

\noindent Choose $g_{\min}$ before the data, from the design or from $\Ghat$ on
the training fold with a stated margin. Revising it upward after seeing $W$ is
the move to avoid, since $g_{\min}$ is what Assumption~\ref{as:pos} needs in
order to bound $G$ from below. Sections~\ref{sec:supptheory}
and~\ref{sec:suppreading} read these results further.

\subsection{Per-cause thresholds}
\label{sec:percause}

The same construction can be run once per
label. For each $\ell\in\{0,\dots,K\}$ let
$\Ical_\ell=\{i\in\Ical: Y_i=\ell\}$ be the label-complete calibration subjects
whose realised label is $\ell$, let $\qhat_\ell$ be \eqref{eq:qhat} computed on
$\Ical_\ell$ alone with the same atom $A$, and put
\begin{equation}\label{eq:mondrian}
S^{\mathrm{M}}(x)\;=\;\{\ell:\ s(x,\ell)\le \qhat_\ell\}.
\end{equation}
This is the class-conditional, or Mondrian, form of the procedure, due in the
uncensored setting to Vovk et al. Its guarantee is per class rather than pooled,
and Corollary~\ref{cor:mondrian} states what censoring does to it.

\begin{corollary}[Per-class coverage]\label{cor:mondrian}
Run the per-class procedure \eqref{eq:mondrian} with estimated weights and the
atom $A_{\mathrm{env}}$, with $\Ghat$ floored at $g_{\min}$ and non-increasing as in
Theorem~\ref{thm:main}. Under Assumptions~\ref{as:indep} to~\ref{as:split},
for every label $\ell$ with $\Prob(Y=\ell)>0$, conditionally on
$\mathcal D_{\mathrm{tr}}$ and on the event $\{\Ghat\ge g_{\min}\}$,
\[
\Prob\big(Y\in S^{\mathrm{M}}(X)\,\big|\,Y=\ell,\mathcal D_{\mathrm{tr}}\big)\;\ge\;1-\alpha-\Delta_w^{(\ell)},
\qquad
\Delta_w^{(\ell)}\;=\;\frac{1}{2\kappa_\ell}\,\E\Big[\Big|\frac{G}{\Ghat}-\kappa_\ell\Big|
\,\Big|\,Y=\ell\Big],
\]
with $\kappa_\ell=\E[G/\Ghat\mid Y=\ell]$. The supremum-norm bound of Theorem~\ref{thm:main}
holds with its expectation taken over $X$ given $Y=\ell$, since
$\kappa_\ell\ge\E[G\mid Y=\ell]\ge g_{\min}$.
\end{corollary}

\begin{proof} In Section~\ref{sec:suppproofs}. \end{proof}

\noindent The conditioning is on the \emph{realised} label, not known at
prediction time, so the guarantee is about the subgroup of subjects who turn out
to have label $\ell$ and not a selection rule. The per-class thresholds inherit the operating condition of Remark~\ref{rem:alpha} one group at a time, so
a class with $W_\ell<A(1-\alpha)/\alpha$ has $\qhat_\ell=+\infty$ and satisfies
the corollary by entering every set.

A pooled coverage figure can sit at nominal while a rare cause is badly missed.
Table~\ref{tab:mondrianbody} makes one cause rare, at $0.15$ of events, and runs on
oracle calibration weights at the uniform constant atom, so it isolates the per-cause effect from the
censoring model. Pooled coverage reads $0.902$ while the decomposition under
that single threshold is $0.920$ for the common cause, $0.994$ for event-free and
$0.118$ for the rare one. Per-cause thresholds move every class to nominal, the
common cause to $0.902$, the rare one to $0.914$ and event-free to $0.903$, at a
mean set size of $2.42$ against $1.84$. Below the boundary of
Remark~\ref{rem:alpha} a class threshold is infinite and its per-cause coverage
is the whole label space, so $W_\ell$ has to be reported with it.

\begin{table}[t]
\centering
\caption{The registry, at the horizon the body discusses. SEER breast competing risks at $\tstar=60$ months, $n=238,057$, $\alpha=0.1$, 5 random splits. Coverage is decomposed by the subject's realised label, first under the single pooled threshold and then under the per-class thresholds of Corollary~\ref{cor:mondrian}. The pooled threshold attains 0.900 overall at a mean set size of 1.08 out of 3 while covering other-cause death 0.262 of the time. Per-class thresholds return every class to nominal at a mean set size of 1.67. $W_\ell$ is the label-complete calibration mass the class carries, which Remark~\ref{rem:alpha} requires to exceed 9.2 at the realised atom 1.0215 for the class threshold to be finite, and every class here clears it. Coverage is the inverse-probability-of-censoring Kaplan--Meier estimate rather than an oracle quantity, since the censoring law is unknown on real data. Section~\ref{sec:realdata} gives the cohort, the two filters, the recodings and the second horizon.}
\label{tab:seerbody}
\begin{tabular}{lrrr}
\toprule
realised label & pooled threshold & per-class thresholds & $W_\ell$ \\
\midrule
cancer death & 0.576 & 0.897 & 8,995 \\
other-cause death & 0.262 & 0.900 & 4,539 \\
event-free & 0.988 & 0.900 & 65,829 \\
\bottomrule
\end{tabular}
\end{table}

The same gap appears on a SEER breast cancer cohort of $238{,}057$ subjects with
two causes of death, in Table~\ref{tab:seerbody} and in full in
Section~\ref{sec:realdata}.
At $\tstar=60$ months the reduction is not under test, only $1.9\%$ of subjects
being censored before the horizon, so what the cohort shows is class imbalance
rather than censoring. Other-cause death is $5.7\%$ of label-complete test
subjects and a single threshold covers it $0.262$ of the time at a pooled
$0.900$, while per-class thresholds give $0.897$, $0.900$ and $0.900$ at a mean
set size of $1.67$ against $1.08$. Membership, a sole registry primary, is
settled during follow-up and removes $31\%$ of the breast records, so the
guarantee refers to that subpopulation and not to a patient at diagnosis. No
guarantee here applies to it numerically either, since censoring depends on the year of
diagnosis and the penalty of Theorem~\ref{thm:main} is not estimable from the
data at hand.

\section{Remarks on the theory}\label{sec:supptheory}

\paragraph{Reading Lemma~\ref{lem:we}'s weight functions.} The weight functions
$(1,\dots,1,\wstar)$ are the hypothesis of Tibshirani et al.'s Definition~1, not
the assignment probabilities they produce. Under their Lemma~3 the case
$(1,\dots,1,w)$ collapses to mass proportional to $w$ on \emph{every} slot, which
is what \eqref{eq:assign} states. The two displays sit one sentence apart in that
lemma and are not two claims about the same object.

\medskip
\noindent\textbf{Ties.} Scores can tie, for instance when the fitted
probabilities are constant on groups of covariates. The implementation does not randomise them. It sorts stably and takes the first score whose cumulative weight reaches the target, so
tied calibration scores all enter at that point. This is not an assumption of
anything proved below. The quantile in \eqref{eq:qhat} is an infimum, so an
inclusive convention can only enlarge the set, and every statement in this paper
is a coverage floor. Randomisation would be needed for a matching upper bound,
which is not claimed.

\begin{remark}[Where the weight sits]\label{rem:testslot}
Two layers have to be kept apart. At the density layer the assignment computation
finds the law of the observed pool proportional to $\wstar$ on the test slot and
to $1$ on each calibration slot, and at the procedure layer \eqref{eq:qhat} is
self-normalised, so a common weight on every slot cancels and amounts to
asserting plain exchangeability, which holds within $\Ical$ and fails against a
population test subject whenever $\wstar$ is not constant. That is the precise
sense in which the unweighted procedure on $\Ical$ is a guarantee about a
different population rather than no guarantee at all. The sign of its bias
depends on the joint law of $(V,\wstar)$, which is what
Section~\ref{sec:simulation} varies. The outcome-dependence of $\wstar$ is
immaterial here, since Definition~1 and the weighted-quantile lemma ask for
absolute continuity, not $X$-measurability.
\end{remark}

\begin{remark}[The atom need not be constant]\label{rem:atomtight}
Lemma~\ref{lem:oracle} is stated for any dominating atom, so the sharp choice
$1/G(\tstar\mid X)$ is covered by it directly, and is no larger than
$1/g_{\min}$. With that atom the mass the
calibration fold needs is near the classical $1/\alpha-1$ rather than
$g_{\min}^{-1}$ times it, so the sharp atom is the better one at known $G$.

No plug-in version of that atom is admissible for this lemma, since under an
estimated censoring distribution $1/\Ghat(\tstar\mid X)$ does not dominate
$1/G(m\mid X)$ (Table~\ref{tab:atomtarget}). Under the theorem's estimated
weights the plug-in envelope is the right atom (Lemma~\ref{lem:atomhat}).
\end{remark}

\section{The full simulation study}\label{sec:suppsim}

This section gives the rest of the synthetic study.

\paragraph{Design constants.} Pools are $2500$ subjects split
in half unless stated, the estimator is a competing-risks forest of $100$ trees in
these designs and $200$ on the registry, figures average $20$ replications
except Table~\ref{tab:rate}, which averages $24$, Tables~\ref{tab:dialsbody}
and~\ref{tab:naturalmisfit}, which average $40$, and the tables whose captions
give their own count, and the censoring rate is
$\lambda_c=\lambda\,c/(1-c)$ for a nominal censoring share $c$. The cause-1
hazard is $\lambda_1e^{x_0}$ divided by the mean of $e^{x_0}$ over the draw, so
subjects within one draw are exchangeable rather than independent, through a
common factor whose fluctuation is of order $n^{-1/2}$. The true score of
Section~\ref{sec:correctfit} uses the same factor. The calibration pool and the
test draw are separate draws with their own factors, so they do not share one law
exactly, and in the Weibull and non-proportional families the censoring rate is
set from the median latent event time of the draw. These designs are therefore
finite-pool illustrations, not exact instances of Assumptions~\ref{as:indep}
to~\ref{as:split}. Table~\ref{tab:iid} reruns the true-score cells of
Table~\ref{tab:oraclescore} with a fixed normaliser, where the draws are
independent and the censoring rate does not depend on them.

Because censoring is covariate-free in these designs, $G(\tstar)$ is fixed by
the design, and Assumption~\ref{as:pos} permits any $g_{\min}$ at or below it,
which is $0.841$, $0.630$, $0.354$ and $0.125$ at censoring $0.2$, $0.4$, $0.6$ and
$0.75$. The designs run at $g_{\min}=0.05$, between $2.5$ and $17$ times looser.
That constant is an analyst's choice. At the atom \eqref{eq:atom} it enters only
through the clip on the calibration weights, which does not bind where those
weights are the oracle ones. Where censoring depends on the covariates, as in
Table~\ref{tab:deltaw}, $G(\tstar\mid X)$ is not a constant.

\paragraph{Two evaluators.} Inside the assumptions coverage is computed against
the generating process's closed-form censoring survival, so the evaluation shares
no estimate with the procedure. Where positivity fails an oracle-weighted
estimate is itself biased, so coverage is computed directly against the latent
label. The one place that disagreement changed a claim is reported below.
\paragraph{Estimating the weights.} Table~\ref{tab:deltaw} drives the
censoring model away from the truth by making the censoring hazard depend on a
covariate with strength $c_s$, and reports the penalty $\Delta_w$ alongside the
floor it implies. The penalty runs from $0.0035$ at $c_s=0$ to $0.2371$ at
$c_s=3$. The unit-atom rows carry no floor, as that table's caption says. Where
there is one, coverage is nowhere below it by more than a standard error. The
closest reading is at $c_s=0$, where $\Delta_w$ is $0.0035$ and the marginal
Kaplan--Meier arm reads $0.896$ (SE $0.003$) against a floor of $0.897$, which is
the theorem holding with almost no room. Everywhere else the floor is loose. The floor at $c_s=3$ is $0.663$ against a realised
$0.868$ for the marginal Kaplan--Meier arm under the capped oracle-weight evaluator, or $0.830$ under
the direct latent-label evaluator of Table~\ref{tab:deltaw}, which scores the
positivity-violating subjects the cap hides. The conclusion holds under either, and the size of the gap does not. The latent-label evaluator is the
stricter of the two.

The rows divide at the point where positivity fails. The population mass with
$G(\tstar\mid X)<g_{\min}$ is estimated from $200{,}000$ draws of the
data-generating process, and it runs $0$, $1\times10^{-5}$, $0.019$, $0.127$ and
$0.200$ across the grid, which is $0$, $2$, $3{,}849$, $25{,}482$ and $40{,}073$
of those draws. Assumption~\ref{as:pos} is almost-sure, so anything above zero is
a violation, but the second entry is two draws and no estimate from this many can
tell it from zero. We therefore treat $c_s\ge1$, where the estimate is
more than five of its own standard errors from zero, as outside the assumptions.
At $c_s=0.5$ the estimate cannot be told from zero, which does not establish the
almost-sure bound, so we report $c_s=0.5$ next to $c_s=0$ without claiming it
satisfies Assumption~\ref{as:pos}. Every arm is valid at $c_s=1$, a property of
that design rather than of the assumption.

\paragraph{Outside positivity.} At $c_s=3$ the conditional-forest arm reads
$0.910$ under the capped evaluator but $0.874$ against the latent label, and the
unweighted arm moves the same way, $0.879$ to $0.841$. The capped evaluator cannot see the fifth of the population that defines
the failure. Outside the assumption the weighting rescues nothing.

\paragraph{The rate.} Under the joint-independence hypothesis of
Corollary~\ref{cor:rate} the penalty decays at a fitted log-log slope of
$-0.497$ over $n$ from $500$ to $32{,}000$ (Table~\ref{tab:rate}). Under the
conditional hypothesis, where a marginal estimator does not converge to
$G(\cdot\mid X)$, it is flat at $0.107$ over the same range. The corollary needs
the branch it names.
\paragraph{Where the crossing is, and how well we locate it.}
Table~\ref{tab:dialsbody} runs the grid at $40$ replications and
Table~\ref{tab:dialsdense} reruns it at $250$ on a six-point $b$ grid, against a
bar of $\lvert t\rvert\ge4$ fixed before the run. Table~\ref{tab:contour} locates
the crossing itself. Resampling the same replications, the interpolated $b^\ast$
runs from $0.319$ to $0.557$. How that
compares with the misspecification an analyst commits without a dial has to be
asked inside one family at a time, because that mistake is not the same size in
all three. The smallest crossing here, $0.319$ under non-proportional hazards at
censoring $0.20$, falls in the family where that mistake is largest, and that
pairing is not separated from zero. The bootstrap
interval at exponential hazards and censoring $0.20$ is $[0.460, 0.575]$, wider
than the grid spacing of $0.1$, which is the condition fixed beforehand for
saying that $250$ replications do not resolve a location. And the contour rises
with the censoring rate in the Weibull and non-proportional families,
$+0.134$ with $[+0.088, +0.178]$ and $+0.185$ with $[+0.145, +0.233]$, but not
in the exponential one, where the paired rise is $+0.038$ with
$[-0.016, +0.093]$.

\paragraph{Further sweeps.} \Supmat\ reports the remaining
sweeps: the ceiling under oracle weights, a set read directly off the fitted model,
the searches for an unweighted failure at a fitted score across censoring rates,
the event-free share and an uninformative score, and a violation of independent
censoring by the latent cause.

\paragraph{The guarantee does not depend on the score.} Validity is
distribution-free in the nonconformity score, and
Table~\ref{tab:scoremisspec} confirms it across three event-time families crossed
with a misspecification dial. The paper's score holds at $0.898$ to $0.902$ in
all fifteen cells even where the fitted model has never seen the covariate the
data generating process uses. Section~\ref{sec:suppextras} gives the design.

\section{A registry cohort}
\label{sec:realdata}

At $\tstar=60$ months only $1.94\%$ of subjects are censored before the horizon,
the censoring weights are flat, and the weighted and unweighted procedures agree
to $0.002$ in coverage, so at that horizon the cohort does not test the
weighting. Section~\ref{sec:percause} uses it for the per-cause gap instead. At
$\tstar=107$ months censoring is administrative and depends on the year of
diagnosis, so a marginal Kaplan--Meier is a misspecified censoring model and a
year-stratified one is available. That makes the cohort a real-data counterpart
of the comparison in Section~\ref{sec:wrongG}. No guarantee of Section~\ref{sec:ceiling} applies here
numerically, since the penalty of Theorem~\ref{thm:main} is a population quantity
in the unknown $G$, so the figures below are measurements.

\medskip
\noindent\textbf{Marginal and year-stratified censoring models at $\tstar=107$.}
The year-stratified model fits one censoring Kaplan--Meier per diagnosis year on
the training fold, with five splits, $\alpha=0.1$ and the same forest and
three-way split as the other registry runs (\texttt{seer-conditional-107.json}).
The two estimates disagree by $0.72$ in sup-norm over $[0,\tstar]$ in the worst
stratum and by $0.27$ averaged over the accrual-year distribution. The largest
calibration weight rises from $1.26$ under the marginal model to $13.40$ under
the stratified one, $1.31\%$ of calibration weights are at or above $5$, and no
weight sits at the $g_{\min}$ clip.

Pooled coverage is $0.9126$ unweighted, $0.9004$ under the marginal weights and
$0.8992$ under the stratified ones, at mean set sizes of $1.35$, $1.28$ and
$1.28$. The three arms share splits, forest and scores, so the gap of $0.0134$
between the first and the last decomposes. Weighting at all accounts for $0.0123$
of it, and correcting the censoring model accounts for $0.0011$, about one
standard error of either arm over the five splits. Scoring all three arms under a
year-stratified coverage estimator instead of the marginal one moves pooled
coverage by at most $0.0006$. At $\tstar=106$, where the largest weight is $7.0$,
the correction leg is $0.0004$. Unlike
Section~\ref{sec:wrongG}, the marginal model is the one closer to nominal here,
and the two models differ by about a thousandth, because their disagreement is
concentrated in a class whose labels are nearly homogeneous.
The censoring estimator removes a failure from the censoring risk set before a
censoring recorded in the same month, as the convention $T\le C$ requires.

Positivity at $\tstar=107$ holds by one month. The minimum accrual-year stratum
censoring survival is $0.077$ against $g_{\min}=0.05$, the smallest $\Ghat$ any
evaluated subject carries is $0.0747$, and at $\tstar=120$ on the full cohort no
follow-up in the 2014 and 2015 strata reaches the horizon, so their censoring
survival is $0$. A per-year minimum averages within the year, so it is
a coarser check than Assumption~\ref{as:pos} requires.

\medskip
\noindent\textbf{Cohort construction.} The cohort is SEER Research Data, 17
Registries, November 2025 submission (2000--2023), breast cancer restricted to
patients for whom it is their sole registry primary and to diagnosis years
2010--2015. The sole-primary restriction makes the two causes well defined and
keeps $991{,}122$ of $1{,}442{,}083$ breast records. Whether a second primary
appears is settled during follow-up, so the analysed population is the patients
whose breast cancer remains their only primary, not patients at diagnosis, and
the restriction is a named mechanism against Assumption~\ref{as:indep}. Two
further filters depend on the outcome. Deaths of unknown cause ($1{,}491$,
$0.60\%$) and records whose survival time is not computed from complete dates
($7{,}357$, $2.98\%$) are removed. The analysed cohort has $238{,}057$ subjects,
of whom $37{,}735$ died of the cancer, $31{,}330$ of other causes and $168{,}992$
were alive or censored at last follow-up. Recoding the removed records in four
ways, from reinstating them under their recorded status to recoding every death
as censored, moves the per-cause coverages under the pooled threshold at
$\tstar=60$ within $[0.557, 0.588]$ for cancer death and $[0.247, 0.275]$ for
other-cause death, and pooled coverage stays in $[0.8996, 0.9001]$. All $4{,}607$
subjects censored before $60$ months are recorded alive, so censoring at that
horizon is loss to follow-up, and Assumption~\ref{as:indep} there is untested.

\subsection{Public cohorts}
\label{sec:public}

Three competing-risks cohorts ship with the R package \texttt{survival}. The
rows of Table~\ref{tab:public} run the same code at the same $\alpha$ over 200
random splits, at horizons fixed from each cohort's censoring distribution
before any coverage was computed.

\begin{table}[t]
\centering
\caption{Three competing-risks cohorts shipped with the R package \texttt{survival}, at $\alpha=0.1$ over 200 random splits. Horizons are the 25th and 50th percentiles of each cohort's observed censoring times, fixed before any coverage was computed. The naive arm drops the weights and the atom. Coverage is the inverse-probability-of-censoring Kaplan--Meier estimate, which shares the marginal censoring model the weighted arm calibrates on, so for that arm it is not a test. Standard errors are over splits of one fixed sample. The last three columns give the least covered cause under the pooled threshold and under per-class thresholds.}
\label{tab:public}
\resizebox{\ifdim\width>\textwidth\textwidth\else\width\fi}{!}{%
\begin{tabular}{lrrrrrrlrr}
\toprule
& & & \multicolumn{2}{c}{coverage} & \multicolumn{2}{c}{set size} & \multicolumn{3}{c}{least covered cause} \\
\cmidrule(lr){4-5}\cmidrule(lr){6-7}\cmidrule(lr){8-10}
cohort & $n$ & $\tstar$ & weighted & naive & weighted & naive & cause & pooled & per-class \\
\midrule
PBC & 399 & 1,441\,d & 0.907 & 0.914 & 1.27 & 1.30 & death & 0.796 & 0.925 \\
PBC & 399 & 2,169\,d & 0.909 & 0.921 & 1.53 & 1.59 & death & 0.876 & 0.929 \\
MGUS & 1,338 & 90\,mo & 0.902 & 0.905 & 1.76 & 1.78 & progression & 0.017 & 0.934 \\
MGUS & 1,338 & 134\,mo & 0.902 & 0.903 & 1.78 & 1.79 & progression & 0.073 & 0.935 \\
FLC & 7,871 & 4,033\,d & 0.900 & 0.912 & 1.80 & 1.93 & neoplastic & 0.334 & 0.903 \\
FLC & 7,871 & 4,607\,d & 0.901 & 0.929 & 2.00 & 2.34 & neoplastic & 0.408 & 0.905 \\
\bottomrule
\end{tabular}
}
\end{table}

Coverage here is estimated under the marginal censoring model the weighted
procedure calibrates on, so the estimates do not establish population coverage.
Under this estimator the unweighted procedure is estimated above nominal in every
cell, by $0.003$ to $0.029$, and the weighted one returns smaller sets in all 6. That model is wrong here.
Censoring depends on the year of diagnosis in mgus2 and of sampling in FLC.
Within quartiles of that year the censoring survival at the later horizon of
mgus2 runs $0.993$, $1.000$, $0.915$ and $0.043$, and of FLC $0.944$, $1.000$,
$0.856$ and $0.000$, so the latest quarter of accrual lies outside
Assumption~\ref{as:pos} in both. PBC, stratified by albumin, stays at or above
$0.366$.

Applying the shrinkage $b$ of Section~\ref{sec:dials} moves the unweighted
procedure and, under this evaluator, not the weighted one. At $b=0.6$ mgus2 at
$134$ months reads $0.8913$, $5.0$ standard errors below nominal, against
$0.9006$ weighted, and PBC at $2{,}169$ days reads $0.8966$ against $0.9034$,
$1.0$ standard error below nominal with $399$ subjects. The paired contrast
between the arms changes sign under the shrinkage, from $+0.0119$ at $b=0$ to
$-0.0067$ at $b=0.6$ on PBC and from $+0.0016$ to $-0.0093$ on mgus2. At $b=0.6$
the mgus2 model predicts an event-free rate of $0.121$ where the complete cases
show $0.302$, so this is a stress test.

Dang et al.\ (2026) report $14.3\%$ coverage on PBC under complete-case conformal
calibration. Their calibration keeps only subjects with an observed event, so the
subjects event-free at the horizon, the group Section~\ref{sec:dials} finds
decisive, are absent from it rather than under-represented, and they predict an
interval on an event time rather than a set of causes.

\section{Deferred proofs}\label{sec:suppproofs}

\paragraph{Proof of the unbiasedness identity of Section~\ref{sec:ceiling}.}
\begin{proof}
Both $Y$ and $m=\min(T,\tstar)$ are functions of $(T,\varepsilon)$. Conditioning
on $(T,\varepsilon,X)$ and using Assumption~\ref{as:indep},
$\E[\Delta\mid T,\varepsilon,X]=\Prob(C\ge m\mid X)=G(m\mid X)=1/\wstar$, so
$\E[\Delta\wstar h\mid T,\varepsilon,X]=h$. Take expectations.
\end{proof}

\paragraph{Proof of Lemma~\ref{lem:we}.}
The assignment claim is that, writing $\{u_0,\dots,u_k\}$ for the conditioned multiset,
\begin{equation}\label{eq:assign}
\Prob\big(\text{test slot}=u_l \,\big|\, \text{multiset},\,|\Ical|=k\big)
\;=\; \frac{\wstar(u_l)}{\sum_{m=0}^{k}\wstar(u_m)} ,
\end{equation}
which is the $(1,\dots,1,\wstar)$ structure by Tibshirani et al.'s Lemma~3.
Write each subject's latent tuple $\zeta=(X,T,\varepsilon)$ with law $\mu$. The
score $V=1-\pi_Y(X)$ and the weight $\wstar=1/G(m\mid X)$ are fixed measurable
functions of $\zeta$, and both factor through the observed reduction
$(X,m,Y)$, so the argument can be read on that pushforward even though
$\varepsilon$ is not observed when $T>\tstar$. Under Assumption~\ref{as:indep},
$\Delta\mid\zeta\sim\mathrm{Bernoulli}(\rho(\zeta))$ independently across
subjects, with selection probability
\[
\rho(\zeta)\;\coloneqq\;\Prob(C\ge m\mid X)\;=\;G(m\mid X)\;=\;1/\wstar(\zeta).
\]
Put $q=\E[\rho(\zeta)]=\Prob(\Delta=1)$.

\emph{Assignment ratio.} Fix $|\Ical|=k$ and the observed multiset, and compare
the labelling that puts $u_0$ in the test slot, with density $f_A$, against the
one that puts $u_1$ there, with density $f_B$. A test slot contributes the plain
density $\mu(u)$. Each retained calibration slot contributes $\rho(u)\mu(u)$. The
$n-k$ censored calibration slots integrate over their unobserved values to the
common factor $(1-q)^{n-k}$, and the slot and order choices are common to both
labellings. Hence
\[
f_A\propto\mu(u_0)\prod_{j=1}^k \rho(u_j)\mu(u_j),\qquad
f_B\propto\mu(u_1)\,\rho(u_0)\mu(u_0)\prod_{j=2}^k \rho(u_j)\mu(u_j),
\]
so $f_A/f_B=\rho(u_1)/\rho(u_0)=\wstar(u_0)/\wstar(u_1)$, which is
\eqref{eq:assign} and the $(1,\dots,1,\wstar)$ structure.

\emph{The random size of $\Ical$.} Under independent
$\mathrm{Bernoulli}(\rho(\zeta_i))$ selection over i.i.d.\ subjects, the size and
the retained values factorise. $|\Ical|\sim\mathrm{Binomial}(n,q)$, and given any
selection pattern with $k$ ones the retained tuples are i.i.d.\ from
$Q_{\mathrm{sel}}\propto \rho\,\mathrm d\mu$, the same law for every pattern,
while each censored subject contributes the value-free constant $1-q$. So
\eqref{eq:assign} holds for every $k$, the case $k=0$ giving $\qhat=+\infty$, and
marginalising over $|\Ical|$ preserves the coverage statement below.

\emph{Coverage.} The structure just established is the weighted split-conformal
instance of Tibshirani et al.\ with test-to-calibration Radon--Nikodym derivative
$\mathrm dP/\mathrm dQ_{\mathrm{sel}}=q\,\wstar\propto\wstar$, and absolute
continuity $P\ll Q_{\mathrm{sel}}$ holds by Assumption~\ref{as:pos} since
$\rho\ge g_{\min}>0$. Their weighted-quantile lemma, with the test point's own
atom at $+\infty$, gives
$\Prob(V_{n+1}\le\qhat\mid\text{multiset})\ge1-\alpha$, and marginalising
finishes. The bound is one-sided and survives ties, because the quantile is an
infimum and ties only inflate coverage.

\paragraph{Proof of Lemma~\ref{lem:oracle}.}
With the oracle test atom $\wstar_{\mathrm{te}}$ the floor is exactly
Lemma~\ref{lem:we}. Replacing that atom by any $A\ge\wstar_{\mathrm{te}}$ only
increases the atom mass, hence weakly increases $\qhat$, since the infimum in
\eqref{eq:qhat} is over a target reached later, and enlarges $S_{\qhat}$. The
comparison is pointwise in the realisation, so the coverage events are nested and
their probabilities ordered, coverage is monotone in $\qhat$ and the floor is
preserved. Nothing here uses $A$ constant. For $A=1/g_{\min}$ the domination is Assumption~\ref{as:pos}, and for $A=1/G(\tstar\mid X)$ it is that
$m=\min(T,\tstar)\le\tstar$ and $G(\cdot\mid X)$ is non-increasing.

\paragraph{Proof of Theorem~\ref{thm:main}.}
\begin{proof}
\emph{The floor.} Lemma~\ref{lem:we} gives, conditionally on $|\Ical|=k$, a
weighted exchangeable vector with the weight on the test slot. Running the
procedure at $\what$ instead of $\wstar$ replaces the oracle Radon--Nikodym
derivative by an estimated one, which is the setting of the total-variation
coupling argument of Lei and Cand\`es (2021, Proposition~1), applied to the
fixed-length vector for each $k$.

That proposition requires the ratio normalised to mean $1$ under the calibration
law. Here the calibration law is $Q_{\mathrm{sel}}$ with
$\mathrm dQ_{\mathrm{sel}}=(\rho/q)\,\mathrm d\mu$, so the true normalised ratio
is $r^\star=\mathrm dP/\mathrm dQ_{\mathrm{sel}}=q\wstar$, which indeed satisfies
$\E_{Q_{\mathrm{sel}}}[q\wstar]=\int\wstar\rho\,\mathrm d\mu=1$. The estimate has
to be normalised the same way. Procedure \eqref{eq:qhat} is self-normalised, so
what it realises is not $q\what$ but
\[
\widehat r\;=\;\frac{\what}{\E_{Q_{\mathrm{sel}}}[\what]}\;=\;\frac{q\what}{\kappa},
\qquad
\kappa\;=\;q\,\E_{Q_{\mathrm{sel}}}[\what]\;=\;\E_\mu\Big[\frac{G}{\Ghat}\Big].
\]
The constant $\kappa$ is a deterministic function of the trained $\Ghat$ and not a
sampling fluctuation of order $|\Ical|^{-1/2}$, which is why it cannot be absorbed
into $q$. Proposition~1 then gives the floor with
\[
\Delta_w=\tfrac12\E_{Q_{\mathrm{sel}}}\big|\widehat r-r^\star\big|
=\tfrac12\E_{Q_{\mathrm{sel}}}\big|q\what/\kappa-q\wstar\big|
=\tfrac{1}{2\kappa}\E_\mu\big|G/\Ghat-\kappa\big|,
\]
where the last equality uses
$\E_{Q_{\mathrm{sel}}}[h]=(1/q)\E_\mu[(1/\wstar)h]=(1/q)\E_\mu[G\,h]$. No
discrete-rank term appears, that slack being upper-sided only, and marginalising
over $|\Ical|$ preserves the statement as in Lemma~\ref{lem:we}. That proposition's
exact step is taken at the tilted law under which $\widehat r$ is correct, so the
object it places on the test slot is $\widehat r$ and the atom must dominate
$\what_{\mathrm{te}}$ rather than $\wstar_{\mathrm{te}}$. Lemma~\ref{lem:atomhat} records this, and the atom \eqref{eq:atom} the theorem is stated at qualifies by
its part (i), the constant $w_\infty$ by its part (ii).

\emph{From weight error to censoring error.} Centre and split,
\[
\E_\mu|G/\Ghat-\kappa|\;\le\;\E_\mu|G/\Ghat-1|+|1-\kappa|\;\le\;2\,\E_\mu|G/\Ghat-1|,
\]
using $|1-\kappa|=|\E_\mu[1-G/\Ghat]|\le\E_\mu|1-G/\Ghat|$. On the conditioning event
and for $m\le\tstar$,
$|G/\Ghat-1|=|G-\Ghat|/\Ghat\le g_{\min}^{-1}\lVert\Ghat-G\rVert_{\infty,\tstar}$,
so $\Delta_w\le \kappa^{-1}g_{\min}^{-1}\E\lVert\Ghat-G\rVert_{\infty,\tstar}$. Finally
$\Ghat\le1$, as a survival estimator, together with $G\ge g_{\min}$ from
Assumption~\ref{as:pos}, gives $\kappa\ge\E_\mu[G]\ge g_{\min}$ and the stated bound.
\end{proof}

\paragraph{Proof of Corollary~\ref{cor:rate}.}
\begin{proof}
Under joint independence $G(u\mid X)=G(u)$ and the marginal Kaplan--Meier
estimator does not depend on $X$, so $\sup_{u\le\tstar}|\Ghat(u)-G(u)|$ is a
constant given $\mathcal D_{\mathrm{tr}}$ and the expectation in the bound of
Theorem~\ref{thm:main} has nothing random left inside it, since $\Delta_w\le g_{\min}^{-2}\sup_{u\le\tstar}|\Ghat(u)-G(u)|$. Replacing $\Ghat$ by
$\max(\Ghat,g_{\min})$, which is what the implementation calibrates with, puts the
event of Theorem~\ref{thm:main} in force by construction and does not loosen the
bound, since $G\ge g_{\min}$ by Assumption~\ref{as:pos} makes the truncation
pointwise no further from $G$ than $\Ghat$ is. Gill (1980) gives
$\sup_{u\le\tstar}|\Ghat(u)-G(u)|=O_P(|\mathcal D_{\mathrm{tr}}|^{-1/2})$ for the
Kaplan--Meier estimator of the censoring distribution on $[0,\tstar]$, which needs
the risk set at $\tstar$ to be non-degenerate, that is $\Prob(\min(T,C)\ge\tstar)>0$.
That holds because $C$ is independent of $T$, $G(\tstar)\ge g_{\min}$ and
$\Prob(T>\tstar)>0$ by hypothesis, so
$\Prob(\min(T,C)\ge\tstar)\ge\Prob(T>\tstar)\,G(\tstar)>0$. The rate is in the size of the training fold, on which $\Ghat$ is
fitted, and not in the calibration count of Lemma~\ref{lem:we}.
\end{proof}

\paragraph{Proof of Corollary~\ref{cor:mondrian}.}
\begin{proof}
Condition on $Y=\ell$ throughout. The calibration subjects with $Y_i=\ell$ and a
test subject with $Y=\ell$ are i.i.d.\ draws from $\mu_\ell=\mu(\cdot\mid Y=\ell)$,
because $Y$ is a function of the latent tuple $\zeta$ and the subjects are i.i.d.
Selection is unchanged by the conditioning. Given $\zeta$, $\Delta$ is
$\mathrm{Bernoulli}(\rho(\zeta))$ independently across subjects by
Assumption~\ref{as:indep}, and $\rho(\zeta)=G(m\mid X)\ge g_{\min}$ by
Assumption~\ref{as:pos}, both of which hold pointwise in $\zeta$ and therefore on
$\mu_\ell$. Lemma~\ref{lem:we} applies verbatim with $\mu$ replaced by $\mu_\ell$,
giving weighted exchangeability of $\Ical_\ell$ together with the test subject
under the same weight functions, and the argument of Theorem~\ref{thm:main}
follows with every expectation taken under $\mu_\ell$.
\end{proof}

\paragraph{Proof of Lemma~\ref{lem:atomhat}.}
\begin{proof}
The two lemmas differ in which procedure carries the exact floor, so they differ in
what the atom substitutes for. In Lemma~\ref{lem:oracle} the floor is
Lemma~\ref{lem:we}, whose weighted-quantile step puts $\wstar_{\mathrm{te}}$ on the
test slot, and the atom replaces that. In Theorem~\ref{thm:main} the floor is not
transported from that procedure. Lei and Cand\`es prove their lower bound by
introducing the tilted law $\widetilde P$ with
$\mathrm d\widetilde P=\widehat r\,\mathrm dQ_{\mathrm{sel}}$, under which the
estimate is the true ratio, applying the weighted-quantile inequality \emph{there},
and only then transporting to $P$ at cost
$d_{\mathrm{TV}}(P,\widetilde P)=\tfrac12\E_{Q_{\mathrm{sel}}}|\widehat r-r^\star|
=\Delta_w$. Read on the multiset of Lemma~\ref{lem:we}, the exact step is
$\sum_l\widehat p_l\1\{E_l\}\ge1-\alpha$ with
$\widehat p_l\propto\what(u_l)$, and the transport replaces $\widehat p$ by
$p^\star_l\propto\wstar(u_l)$ at cost their total variation. So the slot the atom
stands in for carries $\what_{\mathrm{te}}$ and not $\wstar_{\mathrm{te}}$, and normalisation is common to every slot and cancels. There is no route to
$\sum_l p^\star_l\1\{E_l\}\ge1-\alpha$ directly, because $E_l$ is decided by a
threshold computed at $\what$. Atom mass
only inflates the quantile, so any $A\ge\what_{\mathrm{te}}$ preserves the floor,
exactly as in Lemma~\ref{lem:oracle}.

For (i), $m=\min(T,\tstar)\le\tstar$
and $\Ghat(\cdot\mid X)$ is non-increasing, so
$\max\{\Ghat(\tstar\mid X),g_{\min}\}\le\max\{\Ghat(m\mid X),g_{\min}\}$, and the same $\Ghat$ supplies both sides, which is what makes the inequality exact rather than
statistical. An envelope built from a \emph{different} estimator than the one
supplying the weights carries no such guarantee. For (ii), the implementation
calibrates with $1/\max(\Ghat,g_{\min})$, so $\what\le1/g_{\min}$ pointwise with
no condition. Corollary~\ref{cor:mondrian}
inherits the statement, its proof taking every expectation under $\mu_\ell$ and
touching the atom only through Lemma~\ref{lem:we}.
\end{proof}

\section{Reading the numbers}\label{sec:suppreading}

\begin{remark}[What pushes the weighted procedure up]\label{rem:conservative}
Two forces push coverage up. The atom in \eqref{eq:qhat} adds mass
$p_\infty=A/(\sum_j\what_j+A)$ before the quantile target is reached, so $\qhat$
is inflated, and that mass grows as the calibration fold gets smaller. And
Theorem~\ref{thm:main} is one-sided. With a correct censoring model and folds of
the size used here the push is small: Section~\ref{sec:simulation} finds the
weighted procedure within $0.012$ of nominal on either side, with $12$ of the
$24$ weighted cells of Table~\ref{tab:dialsbody} below $0.900$ before rounding, so neither
direction is evidence for or against exactness.
\end{remark}

\paragraph{Checking the operating condition.} The condition of
Remark~\ref{rem:alpha} can be checked before any set is formed. Fix $g_{\min}$
first, below the minimum of $\Ghat(\tstar\mid X)$ over the covariate support with a
margin, then compute $W=\sum_{i\in\Ical}\what_i$ on the calibration fold and
compare it with $A(1-\alpha)/\alpha$ for the atom the procedure carries. At
\eqref{eq:atom} that boundary runs $2.8$ to $155.4$ over the sweep of
Figure~\ref{fig:alpha}, and it is $9.2$ on the registry cohort. If $W$ falls
short, enlarge the calibration fold, or in the class-conditional version report
the causes below the boundary as uninformative. Raising $g_{\min}$ after seeing
$W$ satisfies the condition on paper and removes the floor it protects.

\paragraph{Selection and the direction of the bias.} At censoring $0.75$ the
mean score among label-complete test subjects is $0.649$ against $0.436$ over all
test subjects, so label-completeness selects the harder subjects and a threshold
calibrated for the population is too small for that subgroup. At censoring $0.2$ the same gap is $0.548$ against $0.534$ and the
coverage gap closes with it.

That mean shift is a guide to the direction and not what decides it. The
unweighted threshold is the $(1-\alpha)$-quantile of the label-complete scores,
so the procedure under-covers exactly when the population puts less than
$1-\alpha$ of its mass at or below that point, which depends on the joint law of
$(V,\wstar)$ (Remark~\ref{rem:testslot}) and not on a summary of it. Three
covariate groups make the gap concrete (\texttt{e56\_mean\_shift.py}). The
first, of mass $0.2$, fails early from cause $1$, so its score is $0$, and is
selected with probability $0.05$. The second, of mass $0.2$, fails early from
either cause with probability $1/2$, scores $0.5$ and is never censored. The
third, of mass $0.6$, has $\pi=(0.8,0.2,0)$ and is selected with probability
$0.05$. Every selection probability is at least $g_{\min}=0.05$ and the score is
the truth. Label-complete subjects are harder on average, with mean score
$0.457$ against $0.292$, yet selection piles their mass on the second group's
$0.5$, the unweighted threshold falls from the population's $0.8$ to $0.5$, and
coverage over the population is $0.88$ against $0.90$. The other experiments
agree. In the exact search of Appendix~\ref{sec:suppcost}, a downward shift of
the score distribution under selection accompanies $2{,}464$ of the $2{,}638$
failing (law, $\alpha$) pairs, the cruder proxy that the most probable label is also the earliest
holds in $1{,}940$, and the rule that the sign of the shift is opposite to the
correlation between censoring hazard and expected nonconformity holds at
$\gamma=+3$ and fails at $\gamma=-3$ (Table~\ref{tab:boundedcens}).

\section{Further detail}\label{sec:suppextras}

\paragraph{The score-misspecification design.} Table~\ref{tab:scoremisspec} crosses
three event-time families with a misspecification dial $\lambda$. The families
are exponential, Weibull with the two causes given qualitatively different
baseline hazards, and one where the cause-1 covariate effect reverses at a knot. The dial
corrupts what the model is trained on, replacing the informative covariate by
$\sqrt{1-\lambda}\,x_0+\sqrt{\lambda}\,z$ with $z$ independent noise of the same
variance, and then applies the fitted model to the true covariates at calibration
and test time. At $\lambda=1$ the model has never seen the covariate the data
generating process uses. The paper's score holds at $0.898$ to $0.902$ across all
fifteen cells, and an adaptive-set score with randomisation holds at $0.900$ to
$0.904$. The same score without randomisation returns the whole label space in
every cell, at coverage $1.000$ and set size $3.00$, a low-cardinality degeneracy of
non-randomised adaptive sets and not a property of censoring.

\section{A law where unweighted calibration fails at the true score}\label{sec:suppcost}

The label-complete calibration subjects and a label-complete test subject are
exchangeable, so the unweighted procedure is valid for the label-complete
subpopulation. The weighting changes the population the guarantee is about. We
prefer the population because whether a subject will be followed to their label
turns on a censoring time that has not happened yet, so the label-complete
subgroup cannot be identified at prediction time. The design below shows that
the unweighted procedure can fail over the population inside the assumptions and
at the true score.

Take two equal covariate groups at $\tstar=1$ with $\alpha=0.2$ and
exponential censoring independent of $X$ with $G(\tstar)=0.082$, above the
$g_{\min}=0.05$ the procedure is configured with, so Assumption~\ref{as:pos}
holds. Group $0$ fails early, at a time uniform on $(0.02,0.30)$, always from cause $1$, so its true
$\pi$ is a point mass and its score is $0$. Group $1$'s event time is uniform on
$(0.96,1.04)$ and it fails from cause $1$ with probability $0.8$, so its true
$(\pi_0,\pi_1,\pi_2)$ is $(0.5,0.4,0.1)$ and its $m=\min(T,\tstar)$ is within $0.04$ of the horizon. The
score is the true $\pi$, which is what Theorem~\ref{thm:main} conditions on.

Label-completeness keeps $0.68$ of group $0$ and $0.084$ of group
$1$, so $0.89$ of the calibration sample is the group whose score is $0$, more
than the $1-\alpha=0.8$ the quantile needs, and the unweighted $1-\alpha$ quantile is $0$ with it. Over the population, scored
against the latent label on every test subject, the unweighted procedure covers
$0.499$ against a nominal $0.800$, at a mean set size of $0.50$ of $3$. Below the threshold it returns the empty set, and on group $1$ it covers nothing at all. The
weights are $1/G(m)$, which is between $1.05$ and $2.12$ on group $0$ and between $11.0$ and $12.2$ on group $1$,
and they restore that group's mass, and the weighted procedure covers $0.950$ at
$\qhat=0.600$. What that threshold buys is the point. On group $0$ it returns the
single correct label, and on group $1$, the only subjects whose label at $\tstar$
is in doubt, it returns two labels of three and covers $0.900$. Inside the
assumptions, the unweighted procedure fails and the weighted one is informative
about exactly the subjects the failure concerns.

The size-$2$ set is available only because group $1$'s two causes are not equally
likely. With probability $1/2$ each the two causes tie in score, no threshold
admits one without the other, and the weighted set on group $1$ is the whole label
space. That version is the control, over two hundred replications.

The figures above use the closed-form $G$, so on their own they exhibit
Lemma~\ref{lem:oracle} rather than Theorem~\ref{thm:main}. Refitting the weights
from a Kaplan--Meier estimate on an independent training fold, which is the
regime the theorem covers, the exhibit survives, at coverage $0.948$ and a group-$1$ set size of $1.990$ against $0.950$ and $2.000$ at the oracle weights.

The mechanism is the one Remark~\ref{rem:testslot} names. Here selection
over-represents the zero-score group, so the label-complete quantile sits below
the population's and the unweighted threshold is too small. The same comparison drives Section~\ref{sec:correctfit},
where the true score under-covers once few subjects are event-free at the
horizon, and Section~\ref{sec:wrongG}, where it fails between two regions. The
guarantee is what covers both signs.

The failure is not confined to constructed corners. An exhaustive exact search
over $5{,}985$ three-label laws with fixed failure
times, the true score and covariate-free censoring, each at four levels $\alpha$, finds $2{,}638$ of the $23{,}940$ pairs, from $1{,}645$ laws, in which the
unweighted procedure covers below nominal, the worst at $0.400$ against a nominal
$0.80$ (\texttt{cr-counterexample.py}).

The two-point example below is the same phenomenon at its degenerate end, and the
contrast is the reason the design above needs three labels and a covariate. With two labels, no covariate and a marginal probability below $1-\alpha$, no
non-randomised set short of the whole space exists at any threshold, so the
triviality is a fact about the target and not about anything a procedure chose.
The design above escapes it because three labels and an asymmetric cause split put
a set of size $2$ inside the score space to begin with. With $T\sim\mathrm{Exp}(1)$,
$C\sim\mathrm{Exp}(1.5)$ at $\tstar=1$ and no covariates, label-complete selection
moves the probability of the early-failure label from $0.632$ to $0.817$ and the
unweighted procedure covers $0.632$ against a nominal $0.80$. Its weighted set is
the whole label space, since $P(Y=1)=0.632<0.8$.

\clearpage
\section{Further figures}
\label{sec:s-figures}

\begin{figure}[!h]
\centering
\includegraphics[width=0.8\textwidth]{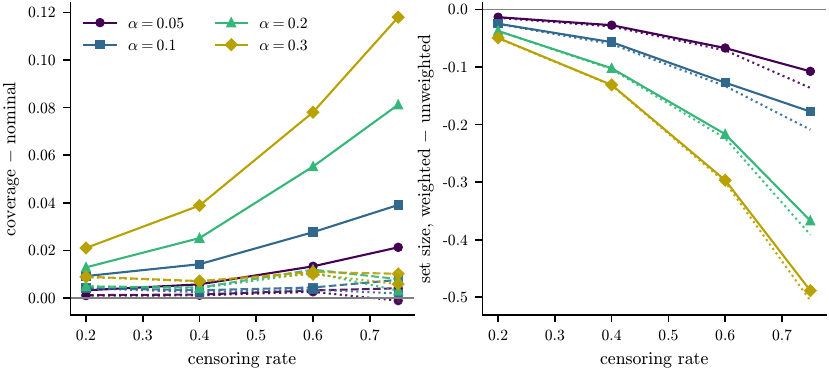}
\caption{The measured advantage, in both of its arguments. Left, deviation from nominal for the unweighted procedure (solid), for the procedure as run at the envelope atom \eqref{eq:atom} (dashed), and at a matched unit atom (dotted). Right, the difference in mean set size against the unweighted procedure for the same two weighted arms, where negative means the weighted sets are smaller. The four colours are the miscoverage levels, keyed in the left panel.}
\label{fig:alpha}
\end{figure}

\begin{figure}[!h]
\centering
\includegraphics[width=0.72\textwidth]{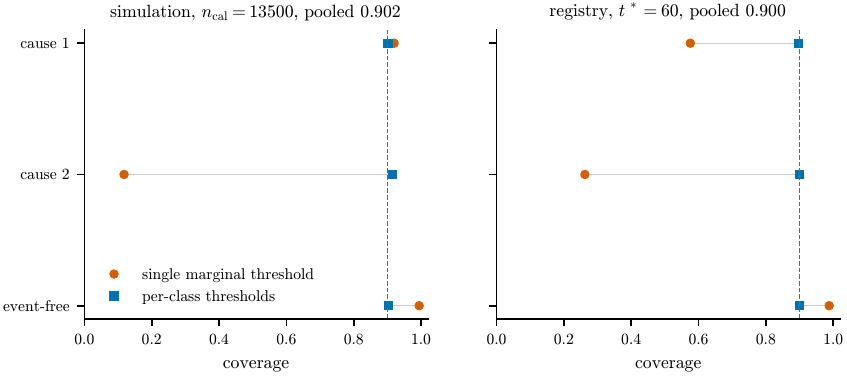}
\caption{A single marginal threshold covers in aggregate while missing a cause, and per-class thresholds repair it. Left, the simulated design of Table~\ref{tab:mondrianbody}. Right, the registry at $\tstar=60$. Circles are coverage under one marginal threshold, squares under per-class thresholds, and the dashed line is the nominal $1-\alpha$.}
\label{fig:percause}
\end{figure}

\clearpage
\section{Further tables}
\label{sec:s-tables}

\begin{table}[t]
\centering
\caption{Per-cause coverage when one cause is rare, at $0.15$ of events, with $K=2$ and $\alpha=0.1$ on an enlarged calibration fold. \emph{marg} is coverage under a single pooled threshold and \emph{per-cause} under per-class thresholds. $W_\ell$ is the label's calibration mass, which must exceed the boundary of Remark~\ref{rem:alpha} for the class threshold to be finite, $p_{\mathrm{eff}}=(1-\alpha)(W_\ell+w_\infty)/W_\ell$ is the level a per-class threshold is calibrated to, and \emph{incl} is the rate at which the label enters a prediction set. The calibration weights are the oracle ones.}
\label{tab:mondrianbody}
\resizebox{\ifdim\width>\textwidth\textwidth\else\width\fi}{!}{%
\begin{tabular}{lrrrrrr}
\toprule
& marg & per-cause & SE & $W_\ell$ & $p_{\mathrm{eff}}$ & incl \\
\midrule
cause 1 & 0.920 & 0.902 & 0.002 & 4601 & 0.904 & 0.73 \\
cause 2 & 0.118 & 0.914 & 0.004 & 1053 & 0.917 & 0.90 \\
event-free & 0.994 & 0.903 & 0.002 & 7861 & 0.902 & 0.79 \\
pooled & 0.902 & & 0.001 & & & \\
mean set size & 1.84 & 2.42 & & & & \\
\bottomrule
\end{tabular}
}
\end{table}

\begin{table}[t]
\centering
\caption{Censoring that depends on the covariate carrying the signal, on bounded support. $X_0$ is uniform on $[-1,1]$, the censoring hazard is proportional to $\exp(\gamma X_0)$, and its level is set so that the smallest $G(t^\ast\mid X)$ is $0.08$, so Assumption~\ref{as:pos} holds by construction. The score is the oracle conditional incidence, and coverage is measured against the latent label. U is unweighted, M uses marginal Kaplan--Meier weights, B uses Kaplan--Meier weights within 10 bins of $X_0$, and O uses the oracle $G$. The control has covariate-free censoring at the same censored fraction and event draws, so the last column is paired. Standard errors are over 200 replications, in units of the last digit.}
\label{tab:boundedcens}
\resizebox{\ifdim\width>\textwidth\textwidth\else\width\fi}{!}{%
\begin{tabular}{rrrrrrrr}
\toprule
$\gamma$ & censored & U & M & B & O & U, control & U, paired \\
\midrule
$-3$ & 0.287 & 0.9060(5) & 0.8988(5) & 0.9006(6) & 0.9007(5) & 0.9079(5) & -0.0020(2) \\
$-2$ & 0.394 & 0.9106(4) & 0.9000(5) & 0.9007(6) & 0.9007(6) & 0.9116(5) & -0.0010(2) \\
$-1$ & 0.572 & 0.9211(4) & 0.9026(5) & 0.9007(6) & 0.9010(6) & 0.9212(5) & -0.0002(2) \\
$+0$ & 0.583 & 0.9220(5) & 0.9013(5) & 0.9011(5) & 0.9012(5) & 0.9220(5) & +0.0000(0) \\
$+1$ & 0.585 & 0.9193(5) & 0.8967(5) & 0.9011(6) & 0.9015(6) & 0.9221(5) & -0.0029(2) \\
$+2$ & 0.364 & 0.9058(5) & 0.8949(5) & 0.9010(5) & 0.9013(5) & 0.9105(5) & -0.0047(2) \\
$+3$ & 0.235 & 0.9008(5) & 0.8946(5) & 0.9009(5) & 0.9011(5) & 0.9063(5) & -0.0055(2) \\
\bottomrule
\end{tabular}
}
\end{table}

\begin{table}[t]
\centering
\caption{A misspecified score under covariate-dependent censoring. The design of Table~\ref{tab:boundedcens}, with the true event-free probability multiplied by $1-b$ and the causes rescaled to match; $b=0$ is that table. Arms and evaluator are as there. Standard errors are over 200 replications, in units of the last digit.}
\label{tab:fourthcell}
\resizebox{\ifdim\width>\textwidth\textwidth\else\width\fi}{!}{%
\begin{tabular}{rrrrrr}
\toprule
$b$ & $\gamma$ & U & M & B & O \\
\midrule
0.0 & $-3$ & 0.9060(5) & 0.8988(5) & 0.9006(6) & 0.9007(5) \\
 & $+0$ & 0.9220(5) & 0.9013(5) & 0.9011(5) & 0.9012(5) \\
 & $+3$ & 0.9008(5) & 0.8946(5) & 0.9009(5) & 0.9011(5) \\
\addlinespace
0.3 & $-3$ & 0.9101(5) & 0.9068(5) & 0.9001(5) & 0.9004(5) \\
 & $+0$ & 0.9142(5) & 0.9005(6) & 0.9005(6) & 0.9005(6) \\
 & $+3$ & 0.8697(6) & 0.8625(6) & 0.8975(6) & 0.9009(6) \\
\addlinespace
0.6 & $-3$ & 0.9198(4) & 0.9269(4) & 0.8996(5) & 0.8999(5) \\
 & $+0$ & 0.8627(7) & 0.8994(7) & 0.8995(7) & 0.8995(6) \\
 & $+3$ & 0.8349(6) & 0.8413(6) & 0.8988(9) & 0.8992(8) \\
\bottomrule
\end{tabular}
}
\end{table}

\begin{table}[t]
\centering
\caption{The true-score cells of Table~\ref{tab:oraclescore} rerun with every subject drawn independently from one law: the cause-1 hazard is normalised by the population value $e^{1/2}$ rather than by each draw's mean, and no forest is fitted. Coverage is inverse-probability weighted under the oracle $G$ over 200 replications with the seeds of Table~\ref{tab:oraclescore}, with standard errors in units of the last printed digit. The share column is the design parameter. The weighted arm uses training-fold Kaplan--Meier weights and the envelope atom \eqref{eq:atom}.}
\label{tab:iid}
\resizebox{\ifdim\width>\textwidth\textwidth\else\width\fi}{!}{%
\begin{tabular}{rrrr}
\toprule
share & censoring & unweighted & weighted \\
\midrule
0.50 & 0.50 & 0.9173(8) & 0.9012(9) \\
0.25 & 0.50 & 0.9008(9) & 0.9027(10) \\
0.15 & 0.50 & 0.8723(13) & 0.9054(14) \\
\addlinespace
0.50 & 0.20 & 0.9053(7) & 0.9005(7) \\
0.50 & 0.40 & 0.9123(8) & 0.9005(8) \\
0.50 & 0.60 & 0.9241(8) & 0.9021(10) \\
0.50 & 0.70 & 0.9327(8) & 0.9035(11) \\
0.50 & 0.75 & 0.9380(9) & 0.9054(13) \\
\bottomrule
\end{tabular}
}
\end{table}

\begin{table}[t]
\centering
\caption{Table~\ref{tab:twobytwo} on real covariates. For each cohort the truth is a cause-specific exponential regression fitted to the cohort, subjects are resampled from its covariate rows, and censoring is synthetic with hazard proportional to $\exp(\gamma z)$, $z$ the rank of the true all-cause hazard, so $G$ is known and $\min G(t^\ast\mid X)\approx0.08$. The horizon is $134$ months for mgus2 and $2169$ days for PBC, and $b$ shrinks the true event-free probability as in Table~\ref{tab:fourthcell}. B uses Kaplan--Meier within $10$ bins of $z$. Standard errors are over 200 replications, in units of the last digit.}
\label{tab:realcov}
\resizebox{\ifdim\width>\textwidth\textwidth\else\width\fi}{!}{%
\begin{tabular}{llrrrrr}
\toprule
cohort & $b$ & $\gamma$ & U & M & B & O \\
\midrule
mgus2 & 0.0 & $-3$ & 0.9077(7) & 0.9083(6) & 0.9004(7) & 0.9016(7) \\
 &  & $+0$ & 0.9036(8) & 0.9008(8) & 0.9021(9) & 0.9009(8) \\
 &  & $+3$ & 0.8885(8) & 0.8871(8) & 0.9018(8) & 0.9007(8) \\
\addlinespace
 & 0.6 & $-3$ & 0.9162(6) & 0.9237(6) & 0.9005(7) & 0.9002(7) \\
 &  & $+0$ & 0.8778(10) & 0.9001(8) & 0.9018(11) & 0.9001(8) \\
 &  & $+3$ & 0.8780(7) & 0.8808(7) & 0.9016(9) & 0.8995(8) \\
\addlinespace
PBC & 0.0 & $-3$ & 0.9043(7) & 0.8970(7) & 0.8998(8) & 0.8995(8) \\
 &  & $+0$ & 0.9203(7) & 0.8997(8) & 0.8996(8) & 0.8998(8) \\
 &  & $+3$ & 0.9057(7) & 0.9009(7) & 0.9003(8) & 0.9000(7) \\
\addlinespace
 & 0.6 & $-3$ & 0.9206(6) & 0.9247(6) & 0.9007(7) & 0.9005(7) \\
 &  & $+0$ & 0.8876(9) & 0.9018(8) & 0.9020(9) & 0.9019(8) \\
 &  & $+3$ & 0.8611(7) & 0.8627(7) & 0.9030(9) & 0.9010(8) \\
\bottomrule
\end{tabular}
}
\end{table}

\begin{table}[t]
\centering
\caption{The estimated-weight penalty $\Delta_w$ under covariate-dependent censoring. Censoring depends on $X_0$ with strength $c_s$ at a fixed realised censoring rate of $0.40$. $P[G<g_{\min}]$ is the population mass outside Assumption~\ref{as:pos}, computed from the generating process, so the rows with $c_s\ge1$ lie outside the assumption and are stress tests. \emph{cov} is the oracle-weighted evaluator, which caps weights at $1/g_{\min}$ and so cannot see the violating subjects. The \emph{direct} columns score every test subject against its latent label, overall and split into the violating stratum (\emph{viol.}) and the rest. The \emph{unit atom} rows are diagnostic decompositions without a coverage floor, and $\Delta_w$ is undefined for the unweighted arm. \emph{pos} is the fraction of label-complete subjects whose true $G(\min(T,t^\ast)\mid X)$ is below $g_{\min}$, and \emph{clip} is the rate at which the calibration $\hat G$ hits the floor. $\kappa=\E[G/\hat G]$ is the constant of Theorem~\ref{thm:main}, and bound$/g_{\min}$ is $\|\hat G-G\|/g_{\min}$, which bounds $\Delta_w$ when $\kappa\ge1$. Without that condition the bound is $1/g_{\min}$ times larger. In the rows with $\kappa<1$ that column is not a bound.}
\label{tab:deltaw}
\footnotesize\resizebox{\ifdim\width>\textwidth\textwidth\else\width\fi}{!}{%
\begin{tabular}{llrrrrrrrrrrrrrrrrr}
\toprule
$c_s$ & arm & censored & $P[G<g_{\min}]$ & cov & SE & direct & SE & viol.\ & rest & $\Delta_w$ & floor & $\kappa$ & $\|\hat G-G\|$ & bound$/g_{\min}$ & pos & clip & size & SE \\
\midrule
0.0 & marginal KM & 0.394 & 0.000 & 0.896 & 0.003 & 0.896 & 0.003 & --- & 0.896 & 0.0035 & 0.897 & 0.999 & 0.0194 & 0.39 & 0.0000 & 0.0000 & 2.26 & 0.010 \\
& conditional forest & & & 0.897 & 0.002 & 0.897 & 0.003 & --- & 0.897 & 0.0497 & 0.850 & 1.023 & 0.1091 & 2.18 & 0.0000 & 0.0000 & 2.26 & 0.010 \\
& marginal KM, unit atom & & & 0.896 & 0.003 & 0.896 & 0.003 & --- & 0.896 & 0.0035 & 0.897 & 0.999 & 0.0194 & 0.39 & 0.0000 & 0.0000 & 2.25 & 0.010 \\
& conditional forest, unit atom & & & 0.895 & 0.003 & 0.895 & 0.003 & --- & 0.895 & 0.0497 & 0.850 & 1.023 & 0.1091 & 2.18 & 0.0000 & 0.0000 & 2.25 & 0.010 \\
& \textbf{unweighted}, unit atom & & & 0.905 & 0.002 & 0.906 & 0.002 & --- & 0.906 & --- & --- & --- & --- & --- & 0.0000 & 0.0000 & 2.30 & 0.009 \\
0.5 & marginal KM & 0.395 & 0.000 & 0.898 & 0.003 & 0.897 & 0.003 & 1.000 & 0.897 & 0.0533 & 0.847 & 0.996 & 0.1092 & 2.18 & 0.0000 & 0.0000 & 2.26 & 0.015 \\
& conditional forest & & & 0.899 & 0.003 & 0.898 & 0.003 & 1.000 & 0.898 & 0.0552 & 0.845 & 1.018 & 0.1206 & 2.41 & 0.0000 & 0.0000 & 2.27 & 0.015 \\
& marginal KM, unit atom & & & 0.897 & 0.003 & 0.897 & 0.003 & 1.000 & 0.897 & 0.0533 & 0.847 & 0.996 & 0.1092 & 2.18 & 0.0000 & 0.0000 & 2.26 & 0.015 \\
& conditional forest, unit atom & & & 0.896 & 0.003 & 0.895 & 0.003 & 1.000 & 0.895 & 0.0552 & 0.845 & 1.018 & 0.1206 & 2.41 & 0.0000 & 0.0000 & 2.25 & 0.015 \\
& \textbf{unweighted}, unit atom & & & 0.907 & 0.003 & 0.906 & 0.003 & 1.000 & 0.906 & --- & --- & --- & --- & --- & 0.0000 & 0.0000 & 2.31 & 0.014 \\
1.0 & marginal KM & 0.397 & 0.019 & 0.898 & 0.004 & 0.897 & 0.004 & 0.928 & 0.896 & 0.1073 & 0.793 & 0.997 & 0.2003 & 4.01 & 0.0001 & 0.0000 & 2.28 & 0.017 \\
& conditional forest & & & 0.905 & 0.004 & 0.904 & 0.004 & 0.930 & 0.903 & 0.0632 & 0.837 & 1.008 & 0.1213 & 2.43 & 0.0001 & 0.0001 & 2.32 & 0.019 \\
& marginal KM, unit atom & & & 0.897 & 0.004 & 0.896 & 0.004 & 0.928 & 0.896 & 0.1073 & 0.793 & 0.997 & 0.2003 & 4.01 & 0.0001 & 0.0000 & 2.27 & 0.017 \\
& conditional forest, unit atom & & & 0.897 & 0.004 & 0.895 & 0.004 & 0.928 & 0.895 & 0.0632 & 0.837 & 1.008 & 0.1213 & 2.43 & 0.0001 & 0.0001 & 2.27 & 0.018 \\
& \textbf{unweighted}, unit atom & & & 0.907 & 0.004 & 0.906 & 0.003 & 0.931 & 0.905 & --- & --- & --- & --- & --- & 0.0001 & 0.0000 & 2.33 & 0.017 \\
2.0 & marginal KM & 0.396 & 0.127 & 0.892 & 0.003 & 0.888 & 0.003 & 0.819 & 0.898 & 0.1913 & 0.709 & 1.008 & 0.3207 & 6.41 & 0.0012 & 0.0000 & 2.26 & 0.011 \\
& conditional forest & & & 0.915 & 0.002 & 0.911 & 0.003 & 0.851 & 0.920 & 0.1006 & 0.799 & 0.971 & 0.1138 & 2.28 & 0.0012 & 0.0003 & 2.38 & 0.016 \\
& marginal KM, unit atom & & & 0.891 & 0.003 & 0.887 & 0.003 & 0.818 & 0.897 & 0.1913 & 0.709 & 1.008 & 0.3207 & 6.41 & 0.0012 & 0.0000 & 2.26 & 0.010 \\
& conditional forest, unit atom & & & 0.901 & 0.003 & 0.897 & 0.003 & 0.832 & 0.907 & 0.1006 & 0.799 & 0.971 & 0.1138 & 2.28 & 0.0012 & 0.0003 & 2.31 & 0.016 \\
& \textbf{unweighted}, unit atom & & & 0.902 & 0.003 & 0.898 & 0.002 & 0.833 & 0.908 & --- & --- & --- & --- & --- & 0.0012 & 0.0000 & 2.31 & 0.011 \\
3.0 & marginal KM & 0.405 & 0.200 & 0.868 & 0.005 & 0.830 & 0.006 & 0.552 & 0.898 & 0.2371 & 0.663 & 1.041 & 0.3889 & 7.78 & 0.0014 & 0.0000 & 2.23 & 0.013 \\
& conditional forest & & & 0.910 & 0.004 & 0.874 & 0.005 & 0.628 & 0.935 & 0.1846 & 0.715 & 1.036 & 0.1406 & 2.81 & 0.0014 & 0.0013 & 2.43 & 0.021 \\
& marginal KM, unit atom & & & 0.867 & 0.005 & 0.829 & 0.006 & 0.551 & 0.897 & 0.2371 & 0.663 & 1.041 & 0.3889 & 7.78 & 0.0014 & 0.0000 & 2.22 & 0.014 \\
& conditional forest, unit atom & & & 0.894 & 0.004 & 0.857 & 0.005 & 0.597 & 0.921 & 0.1846 & 0.715 & 1.036 & 0.1406 & 2.81 & 0.0014 & 0.0013 & 2.35 & 0.018 \\
& \textbf{unweighted}, unit atom & & & 0.879 & 0.004 & 0.841 & 0.006 & 0.570 & 0.908 & --- & --- & --- & --- & --- & 0.0014 & 0.0000 & 2.28 & 0.012 \\
\bottomrule
\end{tabular}
}
\end{table}

{\footnotesize
\begin{longtable}{lrrrrrrr}
\caption{Estimated-weight penalty $\Delta_w$ and its rate. The first block reports the oracle-versus-Kaplan--Meier coverage gap, negligible at every pool size, and the sup-error $\|\hat G-G\|$, whose fitted log-log slope is -0.549. That sup-error is a Kaplan--Meier property and holds whether or not the conformal argument does, so the second block reports $\Delta_w$ itself, which is the quantity Corollary~\ref{cor:rate} makes a claim about. At $c_s=0$, the corollary's own regime, its fitted slope is -0.497 against the asserted $-1/2$. At $c_s=1$ it is -0.002. $\Delta_w$ does not decay at all, and that point still satisfies Assumption~\ref{as:indep} in its covariate-conditional form. What fails is the \emph{marginal} estimator's convergence to $G(\cdot\mid X)$, so the corollary's conclusion needs the joint-independence version of that assumption. The two coverage arms share replication seeds, so the gap's standard error is that of the paired difference rather than of two independent means.}\label{tab:rate} \\
\toprule
$n$ & cov (oracle) / $\Delta_w$ & SE & cov (KM) & SE & gap & SE & $\|\hat G - G\|$ \\
\midrule
\endfirsthead
\multicolumn{8}{l}{\emph{Table \ref{tab:rate}, continued}} \\
\toprule
$n$ & cov (oracle) / $\Delta_w$ & SE & cov (KM) & SE & gap & SE & $\|\hat G - G\|$ \\
\midrule
\endhead
\bottomrule
\endfoot
\multicolumn{8}{l}{\emph{coverage gap and sup-error, marginal censoring}} \\\\
\quad 500 & 0.964 & 0.005 & 0.963 & 0.005 & +0.001 & 0.001 & 0.0577 \\
\quad 1000 & 0.931 & 0.002 & 0.929 & 0.003 & +0.001 & 0.001 & 0.0457 \\
\quad 2000 & 0.911 & 0.003 & 0.911 & 0.003 & -0.000 & 0.000 & 0.0280 \\
\quad 4000 & 0.908 & 0.002 & 0.907 & 0.002 & +0.000 & 0.000 & 0.0186 \\
\quad 8000 & 0.905 & 0.001 & 0.905 & 0.001 & +0.000 & 0.000 & 0.0135 \\
\midrule
\multicolumn{8}{l}{\emph{$\Delta_w$ itself, $c_s=0$, marginal KM consistent; slope -0.497}} \\
\quad 500 & 0.00898 & 0.00130 & & & & & 0.0518 \\
\quad 1000 & 0.00595 & 0.00085 & & & & & 0.0310 \\
\quad 2000 & 0.00495 & 0.00069 & & & & & 0.0261 \\
\quad 2500 & 0.00368 & 0.00039 & & & & & 0.0181 \\
\quad 4000 & 0.00253 & 0.00026 & & & & & 0.0169 \\
\quad 8000 & 0.00260 & 0.00025 & & & & & 0.0131 \\
\quad 16000 & 0.00183 & 0.00017 & & & & & 0.0093 \\
\quad 32000 & 0.00097 & 0.00016 & & & & & 0.0050 \\
\multicolumn{8}{l}{\emph{$\Delta_w$ itself, $c_s=1$, marginal KM inconsistent; slope -0.002}} \\
\quad 500 & 0.10824 & 0.00058 & & & & & 0.2094 \\
\quad 1000 & 0.10742 & 0.00041 & & & & & 0.2030 \\
\quad 2000 & 0.10695 & 0.00031 & & & & & 0.2012 \\
\quad 2500 & 0.10673 & 0.00019 & & & & & 0.2005 \\
\quad 4000 & 0.10690 & 0.00016 & & & & & 0.1995 \\
\quad 8000 & 0.10726 & 0.00017 & & & & & 0.1996 \\
\quad 16000 & 0.10690 & 0.00014 & & & & & 0.1989 \\
\quad 32000 & 0.10695 & 0.00013 & & & & & 0.1988 \\
\end{longtable}
}

\begin{table}[t]
\centering
\caption{A controlled violation of Assumption~\ref{as:pos}. The upper tail of the signal covariate, carrying population mass $p$, is given a censoring hazard that sets $G(t^\ast\mid x)$ to half of $g_{\min}$, with the realised censoring rate held at $40\%$ at every $p$. Coverage is measured directly against the latent label, and the capped IPCW estimate is reported alongside. The fixed top-$30\%$ tail column is a control that scores the same subjects when they receive no special censoring. The last column is the rate at which the clip on the marginal $\hat G$ fires.}
\label{tab:posviolation}
\resizebox{\ifdim\width>\textwidth\textwidth\else\width\fi}{!}{%
\begin{tabular}{rrrrrrrrr}
\toprule
$p$ & realized & direct cov.\ & SE & violating & clean & top-30\% & IPCW cov.\ & clip rate \\
\midrule
0.00 & 0.000 & 0.898 & 0.002 & -- & 0.898 & 0.911 & 0.898 & 0.0000 \\
0.05 & 0.050 & 0.896 & 0.002 & 0.895 & 0.897 & 0.912 & 0.896 & 0.0000 \\
0.15 & 0.149 & 0.894 & 0.003 & 0.876 & 0.897 & 0.904 & 0.893 & 0.0000 \\
0.30 & 0.300 & 0.883 & 0.003 & 0.826 & 0.907 & 0.826 & 0.877 & 0.0000 \\
\bottomrule
\end{tabular}
}
\end{table}

{\footnotesize
\begin{longtable}{llrrrrr}
\caption{Score misspecification. $\lambda$ is the fraction of the informative covariate's variance replaced by noise \emph{in training only}. The fitted model is then applied to the true covariates, so the score is misspecified while the evaluation is not. Three score functions are run on the same fits. Coverage is the column that must not move and set size is the column that should. Validity survives a model trained without the covariate driving the outcome, and the misspecification budget is spent on width. Every (family, $\lambda$, score) cell is listed. The un-randomised APS column degenerates here, and that is an artefact of the small label cardinality rather than a property of the score. With $K+1=3$ labels a set that covers by including nearly all of them is only three labels wide, so the degeneracy is cheap to reach and the column bears on $K=2$ alone.}\label{tab:scoremisspec} \\
\toprule
family & $\lambda$ & score & coverage & SE & mean set size & SE \\
\midrule
\endfirsthead
\multicolumn{7}{l}{\emph{Table \ref{tab:scoremisspec}, continued}} \\
\toprule
family & $\lambda$ & score & coverage & SE & mean set size & SE \\
\midrule
\endhead
\bottomrule
\endfoot
exponential & 0.00 & baseline & 0.902 & 0.002 & 2.28 & 0.011 \\
 & 0.00 & aps & 1.000 & 0.000 & 3.00 & 0.000 \\
 & 0.00 & raps & 0.900 & 0.002 & 2.30 & 0.009 \\
 & 0.25 & baseline & 0.902 & 0.002 & 2.30 & 0.010 \\
 & 0.25 & aps & 1.000 & 0.000 & 3.00 & 0.000 \\
 & 0.25 & raps & 0.902 & 0.002 & 2.33 & 0.010 \\
 & 0.50 & baseline & 0.902 & 0.002 & 2.33 & 0.009 \\
 & 0.50 & aps & 1.000 & 0.000 & 3.00 & 0.000 \\
 & 0.50 & raps & 0.904 & 0.002 & 2.36 & 0.012 \\
 & 0.75 & baseline & 0.901 & 0.002 & 2.38 & 0.010 \\
 & 0.75 & aps & 1.000 & 0.000 & 3.00 & 0.000 \\
 & 0.75 & raps & 0.900 & 0.002 & 2.39 & 0.010 \\
 & 1.00 & baseline & 0.900 & 0.002 & 2.55 & 0.010 \\
 & 1.00 & aps & 1.000 & 0.000 & 3.00 & 0.000 \\
 & 1.00 & raps & 0.900 & 0.002 & 2.55 & 0.009 \\
weibull & 0.00 & baseline & 0.899 & 0.002 & 2.35 & 0.010 \\
 & 0.00 & aps & 1.000 & 0.000 & 3.00 & 0.000 \\
 & 0.00 & raps & 0.900 & 0.002 & 2.40 & 0.009 \\
 & 0.25 & baseline & 0.898 & 0.002 & 2.37 & 0.009 \\
 & 0.25 & aps & 1.000 & 0.000 & 3.00 & 0.000 \\
 & 0.25 & raps & 0.900 & 0.002 & 2.42 & 0.011 \\
 & 0.50 & baseline & 0.900 & 0.002 & 2.42 & 0.009 \\
 & 0.50 & aps & 1.000 & 0.000 & 3.00 & 0.000 \\
 & 0.50 & raps & 0.901 & 0.003 & 2.46 & 0.012 \\
 & 0.75 & baseline & 0.901 & 0.002 & 2.47 & 0.009 \\
 & 0.75 & aps & 1.000 & 0.000 & 3.00 & 0.000 \\
 & 0.75 & raps & 0.901 & 0.002 & 2.50 & 0.011 \\
 & 1.00 & baseline & 0.900 & 0.002 & 2.62 & 0.008 \\
 & 1.00 & aps & 1.000 & 0.000 & 3.00 & 0.000 \\
 & 1.00 & raps & 0.901 & 0.001 & 2.62 & 0.008 \\
non-PH & 0.00 & baseline & 0.902 & 0.002 & 2.39 & 0.008 \\
 & 0.00 & aps & 1.000 & 0.000 & 3.00 & 0.000 \\
 & 0.00 & raps & 0.902 & 0.002 & 2.43 & 0.010 \\
 & 0.25 & baseline & 0.901 & 0.002 & 2.40 & 0.011 \\
 & 0.25 & aps & 1.000 & 0.000 & 3.00 & 0.000 \\
 & 0.25 & raps & 0.901 & 0.002 & 2.43 & 0.008 \\
 & 0.50 & baseline & 0.901 & 0.002 & 2.42 & 0.008 \\
 & 0.50 & aps & 1.000 & 0.000 & 3.00 & 0.000 \\
 & 0.50 & raps & 0.902 & 0.003 & 2.44 & 0.012 \\
 & 0.75 & baseline & 0.898 & 0.002 & 2.42 & 0.009 \\
 & 0.75 & aps & 1.000 & 0.000 & 3.00 & 0.000 \\
 & 0.75 & raps & 0.902 & 0.002 & 2.45 & 0.010 \\
 & 1.00 & baseline & 0.900 & 0.002 & 2.44 & 0.009 \\
 & 1.00 & aps & 1.000 & 0.000 & 3.00 & 0.000 \\
 & 1.00 & raps & 0.903 & 0.002 & 2.46 & 0.010 \\
\end{longtable}
}

\begin{table}[t]
\centering
\caption{The registry decomposition under a second estimator. Every other number in this paper comes from one competing-risk forest, so a magnitude like the per-cause gap cannot be separated from a property of that fit. The second row replaces it with an inverse-censoring-weighted multinomial logistic regression fitted on the label-complete training fold, at the same $t^\ast$, the same splits, the same tie convention and the same evaluator, so the estimator is the only thing that changes. The pooled figure is unchanged and the gap between the event-free majority and the competing cause survives. It is 0.726 under the forest against 0.704 under the logistic fit. Per-class thresholds restore all three labels to nominal under both, at a larger mean set size. The magnitudes move, the phenomenon does not.}
\label{tab:estimator}
\resizebox{\ifdim\width>\textwidth\textwidth\else\width\fi}{!}{%
\begin{tabular}{lrrrrrrrrr}
\toprule
$\pi$ & cov & size & cancer death & other cause & event-free & size & cancer death & other cause & event-free \\
\midrule
competing-risk forest & 0.900 & 1.08 & 0.576 & 0.262 & 0.988 & 1.67 & 0.897 & 0.900 & 0.900 \\
multinomial logistic & 0.900 & 1.13 & 0.550 & 0.286 & 0.990 & 1.77 & 0.899 & 0.898 & 0.899 \\
\bottomrule
\end{tabular}
}
\end{table}

\begin{table}[t]
\centering
\caption{Where the unweighted procedure crosses nominal. $b^\ast$ is the $b$ at which the cell's mean crosses 0.900, linearly interpolated on the grid of Table~\ref{tab:dialsdense}, with intervals from 4000 bootstrap draws over the same replications. An interval wider than the grid spacing of $0.1$ does not resolve a location. The last column is the fraction of draws that cross inside the grid, since the interval is conditional on crossing.}
\label{tab:contour}
\resizebox{\ifdim\width>\textwidth\textwidth\else\width\fi}{!}{%
\begin{tabular}{lrrrrrr}
\toprule
family & $c$ & $b^\ast$ & \multicolumn{2}{c}{95\% interval} & width & crossing rate \\
\midrule
exponential & 0.20 & 0.521 & 0.460 & 0.575 & 0.115 & 0.994 \\
  & 0.40 & 0.510 & 0.486 & 0.533 & 0.047 & 1.000 \\
  & 0.55 & 0.521 & 0.508 & 0.534 & 0.026 & 1.000 \\
  & 0.70 & 0.557 & 0.549 & 0.565 & 0.016 & 1.000 \\
weibull & 0.20 & 0.375 & 0.329 & 0.423 & 0.094 & 1.000 \\
  & 0.40 & 0.379 & 0.350 & 0.407 & 0.056 & 1.000 \\
  & 0.55 & 0.426 & 0.413 & 0.440 & 0.027 & 1.000 \\
  & 0.70 & 0.510 & 0.503 & 0.517 & 0.013 & 1.000 \\
non-PH & 0.20 & 0.319 & 0.268 & 0.360 & 0.091 & 1.000 \\
  & 0.40 & 0.355 & 0.337 & 0.374 & 0.037 & 1.000 \\
  & 0.55 & 0.407 & 0.397 & 0.417 & 0.020 & 1.000 \\
  & 0.70 & 0.504 & 0.495 & 0.510 & 0.014 & 1.000 \\
\bottomrule
\end{tabular}
}
\end{table}

{\setlength{\tabcolsep}{3pt}\begin{table}[t]
\centering
\caption{The design of Table~\ref{tab:dialsbody} at 250 replications on a six-point $b$ grid, unweighted procedure. Coverage is inverse-probability weighted under the oracle $G$. Standard errors are marginal over replications, in units of the last printed digit, and $t$ is against nominal. The bar $|t|\ge4$ was fixed before the run, beyond the Bonferroni threshold for the 72 cells.}
\label{tab:dialsdense}
\resizebox{\ifdim\width>\textwidth\textwidth\else\width\fi}{!}{%
\begin{tabular}{llrrrrrrr}
\toprule
family & $c$ & & $b=0.0$ & $b=0.2$ & $b=0.3$ & $b=0.4$ & $b=0.5$ & $b=0.6$ \\
\midrule
exponential & 0.20 & cov. & 0.9054(5) & 0.9046(5) & 0.9038(5) & 0.9020(6) & 0.9004(6) & 0.8985(6) \\
 & & $t$ & +10.42 & +8.77 & +7.32 & +3.65 & +0.71 & -2.66 \\
  & 0.40 & cov. & 0.9122(5) & 0.9108(6) & 0.9085(6) & 0.9052(6) & 0.9006(7) & 0.8951(6) \\
 & & $t$ & +23.75 & +19.26 & +14.59 & +8.52 & +0.86 & -7.70 \\
  & 0.55 & cov. & 0.9205(5) & 0.9187(6) & 0.9158(6) & 0.9105(7) & 0.9023(7) & 0.8916(8) \\
 & & $t$ & +37.51 & +32.41 & +25.81 & +15.24 & +3.08 & -10.96 \\
  & 0.70 & cov. & 0.9333(6) & 0.9329(6) & 0.9310(6) & 0.9258(7) & 0.9131(9) & 0.8900(11) \\
 & & $t$ & +58.35 & +55.16 & +50.35 & +36.87 & +14.41 & -9.29 \\
weibull & 0.20 & cov. & 0.9040(5) & 0.9030(5) & 0.9017(5) & 0.8994(5) & 0.8977(5) & 0.8958(5) \\
 & & $t$ & +7.81 & +5.68 & +3.14 & -1.09 & -4.50 & -7.61 \\
  & 0.40 & cov. & 0.9101(5) & 0.9068(5) & 0.9030(6) & 0.8992(6) & 0.8940(6) & 0.8893(7) \\
 & & $t$ & +18.80 & +12.43 & +5.26 & -1.38 & -9.53 & -16.34 \\
  & 0.55 & cov. & 0.9176(6) & 0.9139(6) & 0.9099(6) & 0.9026(7) & 0.8926(8) & 0.8796(8) \\
 & & $t$ & +31.64 & +23.48 & +16.11 & +3.83 & -9.49 & -24.01 \\
  & 0.70 & cov. & 0.9278(6) & 0.9281(6) & 0.9269(7) & 0.9204(8) & 0.9031(11) & 0.8712(15) \\
 & & $t$ & +45.84 & +45.84 & +40.94 & +25.79 & +2.91 & -19.42 \\
non-PH & 0.20 & cov. & 0.9043(5) & 0.9022(5) & 0.9005(5) & 0.8980(5) & 0.8953(5) & 0.8935(5) \\
 & & $t$ & +8.30 & +4.13 & +0.89 & -3.75 & -8.91 & -12.36 \\
  & 0.40 & cov. & 0.9111(5) & 0.9074(6) & 0.9036(6) & 0.8971(7) & 0.8884(7) & 0.8812(7) \\
 & & $t$ & +20.26 & +13.33 & +6.15 & -4.49 & -17.16 & -26.83 \\
  & 0.55 & cov. & 0.9191(6) & 0.9154(6) & 0.9104(7) & 0.9012(7) & 0.8855(9) & 0.8663(11) \\
 & & $t$ & +33.95 & +25.24 & +15.84 & +1.60 & -16.17 & -31.09 \\
  & 0.70 & cov. & 0.9294(6) & 0.9284(7) & 0.9259(7) & 0.9195(8) & 0.9015(12) & 0.8600(18) \\
 & & $t$ & +46.03 & +41.52 & +35.88 & +23.14 & +1.26 & -22.04 \\
\bottomrule
\end{tabular}
}
\end{table}

}
\begin{table}[t]
\centering
\caption{Where the unweighted arm's error lands. The calibration fold is unweighted throughout, and the inverse-probability-weighted miss mass is decomposed by the true label. Label 0, event-free at the horizon, is exactly the group carrying $m=\tstar$ and so the lowest selection probability. Two conservation checks. The three group miss masses re-sum to the marginal column to within 1.4e-17, and each group's realised selection rate matches its own mean $G(m)$ to within 0.0032, the largest gap over the groups. Censoring depletes the horizon survivors hardest, and at $c=0.70$ their selection rate is 0.1960 against 0.5604 and 0.5332 for the two causes. At $b=0$ the unweighted arm covers those survivors 0.9995 and carries all of its error on the cause labels. At $b=0.6$ the survivors carry +0.0700 of additional miss mass while the two causes improve by -0.0153 and -0.0117. The last row is the control that keeps this from being circular. $\lambda$ replaces the informative covariate with noise in training, which moves the fitted model visibly, taking cause-2 coverage from 0.8282 to 0.7726, yet it adds -0.0003 to the survivors. A misspecification that leaves the depleted group where it was does not move the marginal arm. The damage also scales with the depletion, the same $b$ nets +0.0165 of extra miss mass at $c=0.40$ and +0.0430 at $c=0.70$.}
\label{tab:landing}
\footnotesize\resizebox{\ifdim\width>\textwidth\textwidth\else\width\fi}{!}{%
\begin{tabular}{llrrrrrrrrrr}
\toprule
$c$ & arm & marginal & \multicolumn{3}{c}{event-free} & \multicolumn{3}{c}{cause 1} & \multicolumn{3}{c}{cause 2} \\
 & & miss & share & miss & cov & share & miss & cov & share & miss & cov \\
\midrule
0.40 & $b=0.0$ & 0.0892 & 0.543 & 0.0029 & 0.9946 & 0.251 & 0.0386 & 0.8463 & 0.206 & 0.0477 & 0.7685 \\
  & $b=0.6$ & 0.1057 & 0.543 & 0.0801 & 0.8523 & 0.251 & 0.0106 & 0.9580 & 0.206 & 0.0150 & 0.9274 \\
  & $\lambda=1.0$ & 0.0873 & 0.543 & 0.0000 & 1.0000 & 0.251 & 0.0248 & 0.9012 & 0.206 & 0.0625 & 0.6967 \\
0.70 & $b=0.0$ & 0.0663 & 0.544 & 0.0003 & 0.9995 & 0.250 & 0.0305 & 0.8779 & 0.207 & 0.0355 & 0.8282 \\
  & $b=0.6$ & 0.1094 & 0.544 & 0.0703 & 0.8705 & 0.250 & 0.0152 & 0.9391 & 0.207 & 0.0239 & 0.8845 \\
  & $\lambda=1.0$ & 0.0666 & 0.544 & 0.0000 & 1.0000 & 0.250 & 0.0197 & 0.9209 & 0.207 & 0.0469 & 0.7726 \\
\bottomrule
\end{tabular}
}
\end{table}

\begin{table}[t]
\centering
\caption{The classic competing-risks mistake, on the exponential family. Each cause's cumulative incidence is estimated as one minus the exponential of the negative of the forest's cause-specific Nelson--Aalen cumulative hazard, which treats the competing event as censoring, and the event-free probability as one minus their sum, clipped and renormalised. The shrink column is the factor the mistake applies to the fitted event-free probability, so $b_{\mathrm{eff}}=1-\text{shrink}$. The two fits share the forest, folds and splits, so the last four columns are paired differences with the paired standard error in units of the last printed digit. Table~\ref{tab:misfitdirect} runs the same mistake in all three families.}
\label{tab:naturalmisfit}
\resizebox{\ifdim\width>\textwidth\textwidth\else\width\fi}{!}{%
\begin{tabular}{rrrrrrrrr}
\toprule
$c$ & comp. & shrink & $b_{\mathrm{eff}}$ & unwtd.\ cov. & \multicolumn{2}{c}{paired, unwtd.} & \multicolumn{2}{c}{paired, wtd.} \\
 & & & & & diff & $t$ & diff & $t$ \\
\midrule
0.40 & 0.2 & 0.9120 & 0.088 & 0.9102 & -0.0016(5) & -2.87 & -0.0001(6) & -0.23 \\
0.40 & 0.4 & 0.8563 & 0.144 & 0.9076 & -0.0041(7) & -5.54 & -0.0009(8) & -1.12 \\
0.40 & 0.6 & 0.8355 & 0.164 & 0.9054 & -0.0055(10) & -5.58 & +0.0004(11) & +0.35 \\
0.40 & 0.8 & 0.8624 & 0.138 & 0.9067 & -0.0031(8) & -4.03 & +0.0005(9) & +0.50 \\
0.70 & 0.2 & 0.9380 & 0.062 & 0.9332 & +0.0006(5) & +1.25 & -0.0005(8) & -0.71 \\
0.70 & 0.4 & 0.8958 & 0.104 & 0.9324 & -0.0016(8) & -2.12 & +0.0007(10) & +0.71 \\
0.70 & 0.6 & 0.8852 & 0.115 & 0.9299 & -0.0029(8) & -3.60 & +0.0010(12) & +0.88 \\
0.70 & 0.8 & 0.9105 & 0.090 & 0.9288 & -0.0022(5) & -4.22 & +0.0002(9) & +0.24 \\
\bottomrule
\end{tabular}
}
\end{table}

{\footnotesize
\begin{longtable}{lrrrrrr}
\caption{The classic competing-risks mistake in three families, and the unweighted arm's coverage under it measured directly. Each cause's cumulative incidence is one minus the exponential of the negative of the forest's cause-specific Nelson--Aalen cumulative hazard, clipped and renormalised, and nothing is dialled. The shrink column is the factor the mistake applies to the fitted event-free probability and $b_{\mathrm{eff}}=1-\text{shrink}$. Coverage carries its Monte Carlo standard error over 100 replications in units of the last digit, and $t$ is against nominal. The arm is below nominal by at least two standard errors in 7 of the 48 designs, 4 of them Weibull, 3 non-proportional and 0 exponential, and by at least four in 4. The lowest is 0.8920, under non-PH hazards at censoring 0.40 and competing fraction 0.6.}\label{tab:misfitdirect} \\
\toprule
family & $c$ & comp. & shrink & $b_{\mathrm{eff}}$ & unwtd.\ cov. & $t$ \\
\midrule
\endfirsthead
\multicolumn{7}{l}{\emph{Table \ref{tab:misfitdirect}, continued}} \\
\toprule
family & $c$ & comp. & shrink & $b_{\mathrm{eff}}$ & unwtd.\ cov. & $t$ \\
\midrule
\endhead
\bottomrule
\endfoot
exponential & 0.20 & 0.2 & 0.9126 & 0.087 & 0.9052(8) & +6.77 \\
 & 0.20 & 0.4 & 0.8508 & 0.149 & 0.9029(9) & +3.32 \\
 & 0.20 & 0.6 & 0.8236 & 0.176 & 0.9028(8) & +3.44 \\
 & 0.20 & 0.8 & 0.8552 & 0.145 & 0.9038(10) & +3.88 \\
 & 0.40 & 0.2 & 0.9129 & 0.087 & 0.9119(8) & +14.26 \\
 & 0.40 & 0.4 & 0.8566 & 0.143 & 0.9090(8) & +11.00 \\
 & 0.40 & 0.6 & 0.8361 & 0.164 & 0.9069(8) & +8.16 \\
 & 0.40 & 0.8 & 0.8613 & 0.139 & 0.9089(9) & +10.03 \\
 & 0.55 & 0.2 & 0.9191 & 0.081 & 0.9212(8) & +26.91 \\
 & 0.55 & 0.4 & 0.8685 & 0.131 & 0.9172(9) & +19.91 \\
 & 0.55 & 0.6 & 0.8517 & 0.148 & 0.9146(8) & +18.09 \\
 & 0.55 & 0.8 & 0.8773 & 0.123 & 0.9157(9) & +16.76 \\
 & 0.70 & 0.2 & 0.9384 & 0.062 & 0.9350(9) & +40.16 \\
 & 0.70 & 0.4 & 0.8969 & 0.103 & 0.9328(9) & +35.58 \\
 & 0.70 & 0.6 & 0.8852 & 0.115 & 0.9302(9) & +32.97 \\
 & 0.70 & 0.8 & 0.9102 & 0.090 & 0.9307(9) & +33.62 \\
weibull & 0.20 & 0.2 & 0.9008 & 0.099 & 0.9063(7) & +9.26 \\
 & 0.20 & 0.4 & 0.7995 & 0.200 & 0.9002(8) & +0.22 \\
 & 0.20 & 0.6 & 0.7536 & 0.246 & 0.8967(9) & -3.87 \\
 & 0.20 & 0.8 & 0.7825 & 0.217 & 0.8985(9) & -1.78 \\
 & 0.40 & 0.2 & 0.9060 & 0.094 & 0.9138(7) & +18.53 \\
 & 0.40 & 0.4 & 0.8141 & 0.186 & 0.9027(8) & +3.18 \\
 & 0.40 & 0.6 & 0.7783 & 0.222 & 0.8952(9) & -5.01 \\
 & 0.40 & 0.8 & 0.8106 & 0.189 & 0.8954(11) & -4.19 \\
 & 0.55 & 0.2 & 0.9171 & 0.083 & 0.9245(8) & +31.48 \\
 & 0.55 & 0.4 & 0.8365 & 0.164 & 0.9119(9) & +13.09 \\
 & 0.55 & 0.6 & 0.8093 & 0.191 & 0.8988(11) & -1.14 \\
 & 0.55 & 0.8 & 0.8425 & 0.157 & 0.8954(12) & -3.80 \\
 & 0.70 & 0.2 & 0.9382 & 0.062 & 0.9392(8) & +51.81 \\
 & 0.70 & 0.4 & 0.8781 & 0.122 & 0.9270(10) & +26.82 \\
 & 0.70 & 0.6 & 0.8626 & 0.137 & 0.9130(11) & +11.36 \\
 & 0.70 & 0.8 & 0.8932 & 0.107 & 0.9067(13) & +5.23 \\
non-PH & 0.20 & 0.2 & 0.8555 & 0.144 & 0.9017(9) & +1.88 \\
 & 0.20 & 0.4 & 0.7302 & 0.270 & 0.8970(8) & -3.86 \\
 & 0.20 & 0.6 & 0.6912 & 0.309 & 0.8949(7) & -7.00 \\
 & 0.20 & 0.8 & 0.7607 & 0.239 & 0.9026(8) & +3.12 \\
 & 0.40 & 0.2 & 0.8768 & 0.123 & 0.9085(9) & +9.06 \\
 & 0.40 & 0.4 & 0.7619 & 0.238 & 0.8985(9) & -1.55 \\
 & 0.40 & 0.6 & 0.7235 & 0.277 & 0.8920(9) & -8.60 \\
 & 0.40 & 0.8 & 0.7769 & 0.223 & 0.9055(9) & +5.82 \\
 & 0.55 & 0.2 & 0.9003 & 0.100 & 0.9171(11) & +15.86 \\
 & 0.55 & 0.4 & 0.7992 & 0.201 & 0.9094(10) & +9.24 \\
 & 0.55 & 0.6 & 0.7648 & 0.235 & 0.8988(9) & -1.34 \\
 & 0.55 & 0.8 & 0.8065 & 0.194 & 0.9103(10) & +10.74 \\
 & 0.70 & 0.2 & 0.9385 & 0.062 & 0.9287(12) & +23.39 \\
 & 0.70 & 0.4 & 0.8693 & 0.131 & 0.9292(11) & +26.47 \\
 & 0.70 & 0.6 & 0.8394 & 0.161 & 0.9195(12) & +16.79 \\
 & 0.70 & 0.8 & 0.8619 & 0.138 & 0.9216(11) & +20.49 \\
\end{longtable}
}

\clearpage
\section{Further simulation results}\label{sec:s-further}

These sweeps complement Appendix~\ref{sec:suppsim} and use its
design constants.

\paragraph{The ceiling.} Under oracle weights the procedure sits at nominal, and
the realisable atom's cost is a property of the constant chosen for it
(Table~\ref{tab:ceiling}).

\paragraph{Reading a set off the fitted model.} Without any calibration, a set
can be read off the fitted cumulative incidence functions by ordering the labels
by $\hat\pi_\ell(x)$ and taking the smallest prefix whose fitted mass reaches
$1-\alpha$. The last row of each block of Table~\ref{tab:naivebody} is that set,
on the same forest, split and test fold as the rows above it. It over-covers only
because with three labels a prefix has to overshoot. Splitting the boundary label
with one uniform draw per subject removes the overshoot, and the set then covers
$0.882$ down to $0.846$ across the censoring grid at $2.21$ to $2.04$ labels,
against $0.904$ to $0.911$ at $2.31$ to $2.39$ labels for the weighted procedure.

\begin{table}[t]
\centering
\caption{The two comparators the introduction sets the paper against. Rows are the procedure at the envelope atom \eqref{eq:atom}, the unweighted procedure on the label-complete subsample, and the plug-in set built from the same fitted cumulative incidence functions with no calibration. All three share the forest, the split and the test fold at each censoring rate, so only the set rule and the weights differ. Table~\ref{tab:naive} carries the two further arms that separate the weighting from the atom, and the caption there reports what a randomized boundary split does to the plug-in set.}
\label{tab:naivebody}
\resizebox{\ifdim\width>\textwidth\textwidth\else\width\fi}{!}{%
\begin{tabular}{llrrrrr}
\toprule
censoring & arm & cov & SE & dev & size & SE \\
\midrule
0.2 & IPCW weights, envelope atom $A_{\mathrm{env}}$ & 0.904 & 0.003 & +0.004 & 2.31 & 0.011 \\
 & unit weights, unit atom & 0.909 & 0.002 & +0.009 & 2.33 & 0.011 \\
 & plug-in CIF set, no calibration & 0.973 & 0.002 & +0.073 & 2.72 & 0.011 \\
0.4 & IPCW weights, envelope atom $A_{\mathrm{env}}$ & 0.904 & 0.002 & +0.004 & 2.31 & 0.009 \\
 & unit weights, unit atom & 0.914 & 0.002 & +0.014 & 2.37 & 0.009 \\
 & plug-in CIF set, no calibration & 0.969 & 0.002 & +0.069 & 2.70 & 0.012 \\
0.5 & IPCW weights, envelope atom $A_{\mathrm{env}}$ & 0.908 & 0.003 & +0.008 & 2.35 & 0.014 \\
 & unit weights, unit atom & 0.923 & 0.002 & +0.023 & 2.42 & 0.012 \\
 & plug-in CIF set, no calibration & 0.965 & 0.002 & +0.065 & 2.67 & 0.014 \\
0.6 & IPCW weights, envelope atom $A_{\mathrm{env}}$ & 0.905 & 0.002 & +0.005 & 2.33 & 0.014 \\
 & unit weights, unit atom & 0.928 & 0.002 & +0.028 & 2.46 & 0.014 \\
 & plug-in CIF set, no calibration & 0.959 & 0.002 & +0.059 & 2.65 & 0.014 \\
0.7 & IPCW weights, envelope atom $A_{\mathrm{env}}$ & 0.911 & 0.002 & +0.011 & 2.37 & 0.015 \\
 & unit weights, unit atom & 0.937 & 0.002 & +0.037 & 2.52 & 0.014 \\
 & plug-in CIF set, no calibration & 0.946 & 0.003 & +0.046 & 2.57 & 0.017 \\
0.75 & IPCW weights, envelope atom $A_{\mathrm{env}}$ & 0.910 & 0.003 & +0.010 & 2.39 & 0.020 \\
 & unit weights, unit atom & 0.939 & 0.003 & +0.039 & 2.55 & 0.017 \\
 & plug-in CIF set, no calibration & 0.936 & 0.003 & +0.036 & 2.52 & 0.014 \\
\bottomrule
\end{tabular}
}
\end{table}

\paragraph{Five searches for an unweighted failure, four at a fitted score.} None
found the unweighted procedure below nominal inside the assumptions. Four score
with the fitted forest: across censoring rates (Table~\ref{tab:naive}), on
the in-assumption rows of covariate-dependent censoring (Table~\ref{tab:deltaw}),
across the event-free share (Table~\ref{tab:share}) and across the miscoverage
level crossed with four censoring rates (Figure~\ref{fig:alpha}). The forest
overstates the event-free probability, which pushes the unweighted procedure up,
so these describe the procedure at that fit rather than at a correct one.
Section~\ref{sec:correctfit} reruns three of them at the true score, where the
share search finds under-coverage. The fifth, on bounded covariate-dependent
censoring, uses the true score and holds, lowest at $0.9008$
(Table~\ref{tab:boundedcens}). Crossing an uninformative score with a minority
event-free label, again at a fitted score, the unweighted procedure covers
$0.911$ against a nominal $0.900$ (Table~\ref{tab:corner}).

\paragraph{When the assumptions fail.} Positivity is handled in
Section~\ref{sec:wrongG} and Table~\ref{tab:posviolation}. For independent
censoring, Table~\ref{tab:violation} multiplies the censoring hazard by
$e^{c_c}$ for subjects whose latent cause is the first one. Pooled coverage stays
within $0.016$ of nominal while cause-1 coverage under a single threshold falls from
$0.821$ to $0.681$, and the unweighted arm tracks the weighted one to within $0.024$
on cause-1. Per-cause thresholds hold cause-1 at $0.932$ to $0.945$
across the sweep and pay in the event-free group, $0.887$ at $c_c=3$. These per-cause
thresholds use the constant atom $1/g_{\min}$. Conditioning on a named cause makes
censoring independent within that cause in this construction, but the shared
marginal censoring estimate is not the right model for each class, and the
event-free group is not a named cause, so these readings are measurements of
this design and not a consequence of Corollary~\ref{cor:mondrian}.

{\footnotesize
\begin{longtable}{llrrrrr}
\caption{Unweighted split conformal against the weighted procedure, decomposed by calibration weights and test atom. The cell with IPCW weights and a unit atom is a diagnostic only, since a unit atom carries no floor, and the three $g_{\min}=0.05$ rows in each block are the uniform-constant sensitivity arm. The last row of each block is the plug-in set, the smallest prefix of labels whose fitted mass reaches $1-\alpha$, with no calibration, on the same forest, split and test fold. Monte-Carlo standard errors are given for coverage and set size.}\label{tab:naive} \\
\toprule
censoring & arm & cov & SE & dev & size & SE \\
\midrule
\endfirsthead
\multicolumn{7}{l}{\emph{Table \ref{tab:naive}, continued}} \\
\toprule
censoring & arm & cov & SE & dev & size & SE \\
\midrule
\endhead
\bottomrule
\endfoot
0.2 & unit weights, unit atom & 0.909 & 0.002 & +0.009 & 2.33 & 0.011 \\
 & unit weights, $1/g_{\min}$ atom & 0.909 & 0.002 & +0.009 & 2.33 & 0.011 \\
 & IPCW weights, unit atom & 0.904 & 0.003 & +0.004 & 2.31 & 0.011 \\
 & IPCW weights, envelope atom $A_{\mathrm{env}}$ & 0.904 & 0.003 & +0.004 & 2.31 & 0.011 \\
 & unit weights, $1/g_{\min}$ atom, $g_{\min}=0.05$ & 0.923 & 0.002 & +0.023 & 2.41 & 0.010 \\
 & IPCW weights, unit atom, $g_{\min}=0.05$ & 0.904 & 0.003 & +0.004 & 2.31 & 0.011 \\
 & IPCW weights, $1/g_{\min}$ atom, $g_{\min}=0.05$ & 0.918 & 0.002 & +0.018 & 2.38 & 0.010 \\
 & plug-in CIF set, no calibration & 0.973 & 0.002 & +0.073 & 2.72 & 0.011 \\
0.4 & unit weights, unit atom & 0.914 & 0.002 & +0.014 & 2.37 & 0.009 \\
 & unit weights, $1/g_{\min}$ atom & 0.915 & 0.002 & +0.015 & 2.37 & 0.008 \\
 & IPCW weights, unit atom & 0.903 & 0.002 & +0.003 & 2.31 & 0.009 \\
 & IPCW weights, envelope atom $A_{\mathrm{env}}$ & 0.904 & 0.002 & +0.004 & 2.31 & 0.009 \\
 & unit weights, $1/g_{\min}$ atom, $g_{\min}=0.05$ & 0.932 & 0.002 & +0.032 & 2.47 & 0.010 \\
 & IPCW weights, unit atom, $g_{\min}=0.05$ & 0.902 & 0.002 & +0.002 & 2.31 & 0.009 \\
 & IPCW weights, $1/g_{\min}$ atom, $g_{\min}=0.05$ & 0.916 & 0.002 & +0.016 & 2.38 & 0.009 \\
 & plug-in CIF set, no calibration & 0.969 & 0.002 & +0.069 & 2.70 & 0.012 \\
0.5 & unit weights, unit atom & 0.923 & 0.002 & +0.023 & 2.42 & 0.012 \\
 & unit weights, $1/g_{\min}$ atom & 0.923 & 0.002 & +0.023 & 2.43 & 0.012 \\
 & IPCW weights, unit atom & 0.907 & 0.003 & +0.007 & 2.34 & 0.014 \\
 & IPCW weights, envelope atom $A_{\mathrm{env}}$ & 0.908 & 0.003 & +0.008 & 2.35 & 0.014 \\
 & unit weights, $1/g_{\min}$ atom, $g_{\min}=0.05$ & 0.941 & 0.002 & +0.041 & 2.53 & 0.012 \\
 & IPCW weights, unit atom, $g_{\min}=0.05$ & 0.907 & 0.003 & +0.007 & 2.34 & 0.014 \\
 & IPCW weights, $1/g_{\min}$ atom, $g_{\min}=0.05$ & 0.919 & 0.002 & +0.019 & 2.41 & 0.012 \\
 & plug-in CIF set, no calibration & 0.965 & 0.002 & +0.065 & 2.67 & 0.014 \\
0.6 & unit weights, unit atom & 0.928 & 0.002 & +0.028 & 2.46 & 0.014 \\
 & unit weights, $1/g_{\min}$ atom & 0.930 & 0.002 & +0.030 & 2.47 & 0.014 \\
 & IPCW weights, unit atom & 0.904 & 0.002 & +0.004 & 2.33 & 0.014 \\
 & IPCW weights, envelope atom $A_{\mathrm{env}}$ & 0.905 & 0.002 & +0.005 & 2.33 & 0.014 \\
 & unit weights, $1/g_{\min}$ atom, $g_{\min}=0.05$ & 0.948 & 0.002 & +0.048 & 2.58 & 0.014 \\
 & IPCW weights, unit atom, $g_{\min}=0.05$ & 0.903 & 0.002 & +0.003 & 2.32 & 0.014 \\
 & IPCW weights, $1/g_{\min}$ atom, $g_{\min}=0.05$ & 0.917 & 0.003 & +0.017 & 2.40 & 0.016 \\
 & plug-in CIF set, no calibration & 0.959 & 0.002 & +0.059 & 2.65 & 0.014 \\
0.7 & unit weights, unit atom & 0.937 & 0.002 & +0.037 & 2.52 & 0.014 \\
 & unit weights, $1/g_{\min}$ atom & 0.943 & 0.002 & +0.043 & 2.56 & 0.015 \\
 & IPCW weights, unit atom & 0.907 & 0.002 & +0.007 & 2.35 & 0.014 \\
 & IPCW weights, envelope atom $A_{\mathrm{env}}$ & 0.911 & 0.002 & +0.011 & 2.37 & 0.015 \\
 & unit weights, $1/g_{\min}$ atom, $g_{\min}=0.05$ & 0.959 & 0.002 & +0.059 & 2.66 & 0.014 \\
 & IPCW weights, unit atom, $g_{\min}=0.05$ & 0.905 & 0.002 & +0.005 & 2.35 & 0.013 \\
 & IPCW weights, $1/g_{\min}$ atom, $g_{\min}=0.05$ & 0.918 & 0.002 & +0.018 & 2.41 & 0.016 \\
 & plug-in CIF set, no calibration & 0.946 & 0.003 & +0.046 & 2.57 & 0.017 \\
0.75 & unit weights, unit atom & 0.939 & 0.003 & +0.039 & 2.55 & 0.017 \\
 & unit weights, $1/g_{\min}$ atom & 0.948 & 0.003 & +0.048 & 2.61 & 0.019 \\
 & IPCW weights, unit atom & 0.905 & 0.003 & +0.005 & 2.36 & 0.021 \\
 & IPCW weights, envelope atom $A_{\mathrm{env}}$ & 0.910 & 0.003 & +0.010 & 2.39 & 0.020 \\
 & unit weights, $1/g_{\min}$ atom, $g_{\min}=0.05$ & 0.969 & 0.002 & +0.069 & 2.73 & 0.015 \\
 & IPCW weights, unit atom, $g_{\min}=0.05$ & 0.902 & 0.004 & +0.002 & 2.35 & 0.023 \\
 & IPCW weights, $1/g_{\min}$ atom, $g_{\min}=0.05$ & 0.914 & 0.004 & +0.014 & 2.41 & 0.023 \\
 & plug-in CIF set, no calibration & 0.936 & 0.003 & +0.036 & 2.52 & 0.014 \\
\end{longtable}
}

\begin{table}[t]
\centering
\caption{The event-free share parameter swept at censoring 0.5 independent of $X$, with a fitted forest score. The parameter fixes $S(t^\ast)$ at zero signal, so with the covariate signal the realised event-free shares are 0.544, 0.322 and 0.224, and $G(t^\ast)$ equals the parameter. At this fitted score the unweighted procedure does not fall below nominal, because the forest overstates the event-free probability; at the true score it does at 0.15 (Table~\ref{tab:oraclescore}). Coverage is the oracle-weighted evaluator.}
\label{tab:share}
\resizebox{\ifdim\width>\textwidth\textwidth\else\width\fi}{!}{%
\begin{tabular}{llrrrrrr}
\toprule
share & $G(t^\ast)$ & arm & cov & SE & dev & size & SE \\
\midrule
0.5 & 0.50 & unit weights, unit atom & 0.923 & 0.002 & +0.023 & 2.42 & 0.012 \\
& & KM weights, envelope atom & 0.907 & 0.003 & +0.007 & 2.34 & 0.014 \\
& & oracle weights, per-point atom & 0.907 & 0.003 & +0.007 & 2.34 & 0.014 \\
0.25 & 0.25 & unit weights, unit atom & 0.919 & 0.003 & +0.019 & 2.49 & 0.014 \\
& & KM weights, envelope atom & 0.905 & 0.003 & +0.005 & 2.43 & 0.014 \\
& & oracle weights, per-point atom & 0.904 & 0.003 & +0.004 & 2.42 & 0.014 \\
0.15 & 0.15 & unit weights, unit atom & 0.916 & 0.003 & +0.016 & 2.48 & 0.015 \\
& & KM weights, envelope atom & 0.906 & 0.003 & +0.006 & 2.44 & 0.016 \\
& & oracle weights, per-point atom & 0.903 & 0.003 & +0.003 & 2.43 & 0.016 \\
\bottomrule
\end{tabular}
}
\end{table}

\begin{table}[t]
\centering
\caption{The corner the other two sweeps each stopped short of. The share sweep of Table~\ref{tab:share} corrupts nothing and the score-misspecification sweep of Table~\ref{tab:scoremisspec} holds the share at one half, so the design where an uninformative fitted score meets a minority event-free label had never been run. It is run here, at $n_{\mathrm{pool}}=2500$ over 20 replications, with outcomes fixed in advance. The mean spread of the fitted probabilities falls from 0.191 at $\lambda=0$ to 0.079 at $\lambda=1$. The unweighted procedure stays above nominal in every cell. Its worst cell against the oracle-weighted arm, which carries the constant atom $1/g_{\min}$ rather than \eqref{eq:atom}, is share 0.15 at $\lambda=1$, where it reads 0.911 against 0.917, a difference of -0.005 against a standard error of 0.004 on the unweighted arm alone. That is the direction the mechanism predicts and it is not a failure. The pre-registration is reproduced with the submission.}
\label{tab:corner}
\resizebox{\ifdim\width>\textwidth\textwidth\else\width\fi}{!}{%
\begin{tabular}{llrrrrrrr}
\toprule
share & $\lambda$ & spread & naive cov & SE & size & oracle-wt.\ cov & SE & size \\
\midrule
0.50 & 0 & 0.181 & 0.923 & 0.002 & 2.42 & 0.919 & 0.002 & 2.41 \\
 & 1 & 0.073 & 0.922 & 0.002 & 2.66 & 0.917 & 0.002 & 2.63 \\
0.25 & 0 & 0.197 & 0.919 & 0.003 & 2.49 & 0.917 & 0.003 & 2.48 \\
 & 1 & 0.081 & 0.917 & 0.003 & 2.76 & 0.913 & 0.004 & 2.75 \\
0.15 & 0 & 0.195 & 0.916 & 0.003 & 2.48 & 0.916 & 0.004 & 2.48 \\
 & 1 & 0.083 & 0.911 & 0.004 & 2.77 & 0.917 & 0.003 & 2.78 \\
\bottomrule
\end{tabular}
}
\end{table}

{\footnotesize
\begin{longtable}{llrrrrr}
\caption{The oracle-weight ceiling. Coverage and mean set size under oracle calibration weights, with the per-point test atom $w^\star_{\mathrm{te}}$ of Lemma~\ref{lem:we}, which needs the unobserved outcome, and with the realisable atom $1/g_{\min}$, both at the largest constant Assumption~\ref{as:pos} permits in each design and at $g_{\min}=0.05$. Coverage is oracle-weighted, so it does not depend on $\hat G$. The second panel sweeps the calibration pool size, which is what governs the atom's share of the mass. Paired differences against the per-point atom carry their standard errors.}\label{tab:ceiling} \\
\toprule
setting & atom & cov & SE & dev & size & SE \\
\midrule
\endfirsthead
\multicolumn{7}{l}{\emph{Table \ref{tab:ceiling}, continued}} \\
\toprule
setting & atom & cov & SE & dev & size & SE \\
\midrule
\endhead
\bottomrule
\endfoot
\multicolumn{7}{l}{\emph{censoring rate}, $n_{\mathrm{pool}}=2500$} \\
\quad 0.2 & $w^\star_{\mathrm{te}}$ & 0.904 & 0.003 & +0.004 & 2.31 & 0.011 \\
\quad 0.2 & mean calibration weight & 0.904 & 0.003 & +0.004 & 2.31 & 0.011 \\
\quad 0.2 & $1/g_{\min}$, design constant & 0.904 & 0.003 & +0.004 & 2.31 & 0.011 \\
\quad 0.2 & $1/g_{\min}$, $g_{\min}=0.05$ & 0.918 & 0.002 & +0.018 & 2.38 & 0.010 \\
\quad 0.4 & $w^\star_{\mathrm{te}}$ & 0.903 & 0.002 & +0.003 & 2.31 & 0.009 \\
\quad 0.4 & mean calibration weight & 0.903 & 0.002 & +0.003 & 2.31 & 0.009 \\
\quad 0.4 & $1/g_{\min}$, design constant & 0.903 & 0.002 & +0.003 & 2.31 & 0.009 \\
\quad 0.4 & $1/g_{\min}$, $g_{\min}=0.05$ & 0.916 & 0.002 & +0.016 & 2.38 & 0.008 \\
\quad 0.6 & $w^\star_{\mathrm{te}}$ & 0.904 & 0.002 & +0.004 & 2.33 & 0.014 \\
\quad 0.6 & mean calibration weight & 0.904 & 0.002 & +0.004 & 2.33 & 0.014 \\
\quad 0.6 & $1/g_{\min}$, design constant & 0.904 & 0.002 & +0.004 & 2.33 & 0.014 \\
\quad 0.6 & $1/g_{\min}$, $g_{\min}=0.05$ & 0.917 & 0.002 & +0.017 & 2.40 & 0.015 \\
\multicolumn{7}{l}{\emph{pool size} (the calibration fold is half), censoring 0.4} \\
\quad 500 & $w^\star_{\mathrm{te}}$ & 0.899 & 0.005 & -0.001 & 2.33 & 0.032 \\
\quad 500 & mean calibration weight & 0.899 & 0.005 & -0.001 & 2.33 & 0.032 \\
\quad 500 & $1/g_{\min}$, design constant & 0.900 & 0.005 & -0.000 & 2.33 & 0.032 \\
\quad 500 & $1/g_{\min}$, $g_{\min}=0.05$ & 0.971 & 0.002 & +0.071 & 2.76 & 0.021 \\
\quad 1000 & $w^\star_{\mathrm{te}}$ & 0.899 & 0.004 & -0.001 & 2.31 & 0.024 \\
\quad 1000 & mean calibration weight & 0.899 & 0.004 & -0.001 & 2.31 & 0.024 \\
\quad 1000 & $1/g_{\min}$, design constant & 0.899 & 0.004 & -0.001 & 2.31 & 0.024 \\
\quad 1000 & $1/g_{\min}$, $g_{\min}=0.05$ & 0.931 & 0.003 & +0.031 & 2.49 & 0.023 \\
\quad 2500 & $w^\star_{\mathrm{te}}$ & 0.903 & 0.002 & +0.003 & 2.31 & 0.009 \\
\quad 2500 & mean calibration weight & 0.903 & 0.002 & +0.003 & 2.31 & 0.009 \\
\quad 2500 & $1/g_{\min}$, design constant & 0.903 & 0.002 & +0.003 & 2.31 & 0.009 \\
\quad 2500 & $1/g_{\min}$, $g_{\min}=0.05$ & 0.916 & 0.002 & +0.016 & 2.38 & 0.008 \\
\end{longtable}
}

\begin{table}[t]
\centering
\caption{Outside Assumption~\ref{as:indep}. The censoring hazard is multiplied by $e^{c_c}$ for subjects whose latent cause is cause~1, so censoring depends on the cause and the theorem claims nothing. Coverage is measured against the latent label with no weights, since an IPCW estimate of coverage is itself biased once independence fails. Because the multiplier depends on the cause alone, conditioning on the cause restores the assumption, so the per-cause columns are not evidence of robustness.}
\label{tab:violation}
\resizebox{\ifdim\width>\textwidth\textwidth\else\width\fi}{!}{%
\begin{tabular}{llrrrrrrrr}
\toprule
$c_c$ & censoring & procedure & pooled & SE & cause 1 & SE & cause 2 & event-free & size \\
\midrule
0 & 0.394 & IPCW, marginal threshold & 0.900 & 0.001 & 0.821 & 0.004 & 0.747 & 0.994 & 2.27 \\
 &  & unweighted, marginal threshold & 0.909 & 0.001 & 0.832 & 0.004 & 0.777 & 0.994 & 2.32 \\
 &  & IPCW, per-cause threshold & 0.925 & 0.001 & 0.932 & 0.002 & 0.939 & 0.917 & 2.54 \\
1 & 0.397 & IPCW, marginal threshold & 0.898 & 0.001 & 0.783 & 0.005 & 0.783 & 0.994 & 2.26 \\
 &  & unweighted, marginal threshold & 0.907 & 0.001 & 0.798 & 0.005 & 0.811 & 0.995 & 2.31 \\
 &  & IPCW, per-cause threshold & 0.922 & 0.002 & 0.938 & 0.002 & 0.941 & 0.907 & 2.55 \\
2 & 0.395 & IPCW, marginal threshold & 0.893 & 0.002 & 0.749 & 0.006 & 0.792 & 0.998 & 2.25 \\
 &  & unweighted, marginal threshold & 0.906 & 0.001 & 0.770 & 0.005 & 0.826 & 0.999 & 2.31 \\
 &  & IPCW, per-cause threshold & 0.915 & 0.002 & 0.940 & 0.003 & 0.933 & 0.897 & 2.54 \\
3 & 0.391 & IPCW, marginal threshold & 0.884 & 0.002 & 0.681 & 0.006 & 0.825 & 1.000 & 2.22 \\
 &  & unweighted, marginal threshold & 0.896 & 0.002 & 0.705 & 0.007 & 0.855 & 1.000 & 2.28 \\
 &  & IPCW, per-cause threshold & 0.910 & 0.002 & 0.945 & 0.003 & 0.928 & 0.887 & 2.55 \\
\bottomrule
\end{tabular}
}
\end{table}

\end{document}